\documentclass[11pt,letterpaper]{amsart}
\usepackage{amsmath,amssymb,amsthm}
\usepackage{enumerate}
\usepackage{placeins}
\DeclareMathOperator{\ADF}{ADF}
\DeclareMathOperator{\AMF}{AMF}
\DeclareMathOperator{\CDF}{CDF}
\DeclareMathOperator{\CMF}{CMF}
\DeclareMathOperator{\PSC}{PSC}
\DeclareMathOperator{\len}{len}
\DeclareMathOperator{\RePart}{Re}

\renewcommand{\phi}{\varphi}
\renewcommand{\Re}{\RePart}

\newcommand{\C}{{\mathbb C}}

\newcommand{\N}{{\mathbb N}}

\newcommand{\R}{{\mathbb R}}
\newcommand{\Z}{{\mathbb Z}}

\newcommand{\conj}[1]{\overline{#1}}

\newcommand{\norm}[2]{\Vert{#1}\Vert_{#2}}
\newcommand{\bnorm}[2]{\Big\Vert{#1}\Big\Vert_{#2}}
\newcommand{\normp}[3]{\norm{#1}{#2}^{#3}}
\newcommand{\bnormp}[3]{\bnorm{#1}{#2}^{#3}}

\newcommand{\normot}[1]{\normp{#1}{1}{2}}
\newcommand{\normtt}[1]{\normp{#1}{2}{2}}
\newcommand{\bnormtt}[1]{\bnormp{#1}{2}{2}}
\newcommand{\normtf}[1]{\normp{#1}{2}{4}}
\newcommand{\normff}[1]{\normp{#1}{4}{4}}

\newcommand{\anrd}[1]{\frac{\normff{#1}}{\normtf{#1}}-1}

\newcommand{\cnrd}[2]{\frac{\normtt{{#1} {#2}}}{\normtt{#1}  \normtt{#2}}}
\newcommand{\ggen}[1]{\langle{#1}\rangle}
\newcommand{\ta}{\widetilde{a}}
\newcommand{\tb}{\widetilde{b}}
\newcommand{\tf}{\widetilde{f}}
\newcommand{\tg}{\widetilde{g}}

\newcommand{\modeg}{(-1)^{\ell-1}}
\newcommand{\zs}{z^{2\ell-1}}
\newcommand{\dzs}{z^{4\ell-2}}
\newcommand{\izs}{z^{1-2\ell}}
\newcommand{\idzs}{z^{2-4\ell}}
\newcommand{\apols}{P_\ell}
\newcommand{\agr}{G_\ell}
\newcommand{\cpols}{\apols\times \apols}
\newcommand{\cgr}{G_{\ell,\ell}}
\newcommand{\stt}[1]{{\small{\tt {#1}}}}
\newcommand{\ttt}[1]{{\tiny{\tt {#1}}}}
\newtheorem{theorem}{Theorem}[section]
\newtheorem{proposition}[theorem]{Proposition}
\newtheorem{lemma}[theorem]{Lemma}
\newtheorem{corollary}[theorem]{Corollary}
\theoremstyle{remark}
\newtheorem{remark}[theorem]{Remark}
\title{Crosscorrelation of Rudin-Shapiro-Like Polynomials}

\author{Daniel J.~Katz}
\address{Department of Mathematics, California State University, Northridge, \: United States}
\author{Sangman Lee}
\author{Stanislav A.~Trunov}
\thanks{This paper is based on work of the three authors supported in part by the National Science Foundation under Grant DMS 1500856.}
\date{13 July 2018}
\begin{document}
\begin{abstract}
We consider the class of Rudin-Shapiro-like polynomials, whose $L^4$ norms on the complex unit circle were studied by Borwein and Mossinghoff.  The polynomial $f(z)=f_0+f_1 z + \cdots + f_d z^d$ is identified with the sequence $(f_0,f_1,\ldots,f_d)$ of its coefficients.   From the $L^4$ norm of a polynomial, one can easily calculate the autocorrelation merit factor of its associated sequence, and conversely.  In this paper, we study the crosscorrelation properties of pairs of sequences associated to Rudin-Shapiro-like polynomials.  We find an explicit formula for the crosscorrelation merit factor.  A computer search is then used to find pairs of Rudin-Shapiro-like polynomials whose autocorrelation and crosscorrelation merit factors are simultaneously high.  Pursley and Sarwate proved a bound that limits how good this combined autocorrelation and crosscorrelation performance can be.  We find infinite families of polynomials whose performance approaches quite close to this fundamental limit.
\end{abstract}
\maketitle

\section{Introduction}
This paper concerns the discovery of Rudin-Shapiro-like polynomials that have exceptionally good correlation properties.  Shapiro \cite{Shapiro} recursively constructed a family of polynomials with coefficients in $\{-1,1\}$ that are flat on the complex unit circle: the ratio of their $L^\infty$ to $L^2$ norm never exceeds $\sqrt{2}$. 
Around the same time, Golay \cite{Golay-1951} constructed binary sequences following the same recursion as the coefficients of Shapiro's polynomials.  Shapiro's polynomials were subsequently rediscovered by Rudin \cite{Rudin}.
Littlewood used the $L^4$ norm on the complex unit circle in his investigation \cite{Littlewood-1966} into the flatness of polynomials with coefficients in $\{-1,1\}$, which are now known as {\it Littlewood polynomials}.
He calculated the ratio of $L^4$ to $L^2$ norm of the Rudin-Shapiro polynomials in \cite[Problem 19]{Littlewood-1968}.

Golay \cite{Golay-1972} independently developed the {\it merit factor}, a normalized average of the mean squared magnitude of the aperiodic autocorrelation of sequences used in remote sensing and communications networks.
Eventually it was discovered that determining the $L^4$ norm on the unit circle of a polynomial is tantamount to determining Golay's merit factor for the sequence of coefficients of that polynomial (see \cite[eq.~(4.1)]{Hoholdt-Jensen}).
We shall soon make precise this connection between the analytic behavior of polynomials on the complex unit circle and the correlation behavior of their associated sequences.

Inspired by the work of Littlewood on the Rudin-Shapiro polynomials, Borwein and Mossinghoff \cite{Borwein-Mossinghoff} recursively define sequences $f_0(z), f_1(z), \ldots$ of polynomials, where $f_0(z)$ is any Littlewood polynomial and the rest of the polynomials are obtained via a recursion of the form
\begin{equation}\label{Francis}
f_{n+1}(z)=f_n(z) + \sigma_n z^{1+\deg f_n} f_n^\dag(-z),
\end{equation}
where $\sigma_n \in \{-1,1\}$ represents an arbitrary sign that can differ at each stage of the recursion, and where for any polynomial $a(z)=a_0 + a_1 z + \cdots + a_d z^d \in \C[z]$ of degree $d$, the polynomial $a^\dag(z)$ denotes the {\it conjugate reciprocal polynomial of $a(z)$}, which is $\conj{a_d} + \conj{a_{d-1}} z + \cdots + \conj{a_0} z^d$.\footnote{In fact, Borwein-Mossinghoff use the {\it reciprocal polynomial} $a^*(z)=a_d + a_{d-1} z + \cdots + a_0 z^d$, but since they are working with polynomials with real coefficients, this is the same as $a^\dag(z)$.  In this paper, it was found that using $a^\dag(z)$ instead of $a^*(z)$ gives the natural generalization of their recursion for polynomials with non-real coefficients.} If one sets $f_0(z)=1$, $\sigma_0=1$, and $\sigma_n=(-1)^{n+1}$ for all positive $n$, then $f_0, f_1, f_2, \ldots$ become Shapiro's original polynomials \cite[Theorem 5(ii)]{Shapiro}.

In this paper, we relax the condition that the initial polynomial $f_0(z)$ be a Littlewood polynomial.  We allow $f_0(z)$ to be a polynomial in $\C[z]$, and only impose the condition that $f_0(z)$ have a nonzero constant coefficient.  This ensures that $f_0^\dag(z)$ has the same degree as $f_0(z)$, so that when we construct $f_0(z),f_1(z),\ldots$ via recursion \eqref{Francis}, a straightforward induction shows that every $f_n$ has a nonzero constant coefficient and
\begin{equation}\label{George}
1+\deg f_n=2^n (1+\deg f_0)
\end{equation}
for every $n$.
We call the sequence $\sigma=\sigma_0,\sigma_1,\ldots$ of numbers in $\{-1,1\}$ that occur in our recursion the {\it sign sequence} for that recursion.
We call $f_0(z)$ the {\it seed}, and the sequence $f_0(z), f_1(z), \ldots$ of polynomials obtained from the seed by applying the recursion is called the {\it stem} associated to seed $f_0(z)$ and sign sequence $\sigma$.
Any stem obtained from a seed $f_0(z) \in \C[z]$ with nonzero constant coefficient is also called a {\it sequence of Rudin-Shapiro-like polynomials}.

Borwein and Mossinghoff \cite{Borwein-Mossinghoff} study the $L^4$ norm of Rudin-Shapiro-like Littlewood polynomials on the complex unit circle, or equivalently, the autocorrelation merit factor of these polynomials.
We now describe the correspondence between polynomials and sequences, and the relation of $L^p$ norms on the complex unit circle to correlation.

In this paper, by a {\it sequence of length $\ell$} we mean some $(a_0,a_1,\ldots,a_{\ell-1}) \in \C^\ell$.
Most researchers are especially interested in sequences that are {\it unimodular}, that is, whose terms are all of unit magnitude, and are most of all interested in sequences that are {\it binary}, that is, whose terms lie in $\{-1,1\}$.
We always identify the polynomial $a(z)=a_0+a_1 z+ \cdots + a_{\ell-1} z^{\ell-1} \in \C[z]$ of degree $\ell-1$ with the sequence $a=(a_0,a_1,\ldots,a_{\ell-1})$ of length $\ell$.
With this identification, binary sequences correspond to Littlewood polynomials.
Usually it is easier to work with sequence length rather than polynomial degree, and using our identification of sequences with polynomials, we define the {\it length} of a polynomial to be the length of the sequence associated with the polynomial, that is, $\len a=1+\deg a$.
Notice that when we have a stem $f_0(z),f_1(z),\ldots$ generated from a seed $f_0(z) \in \C[z]$ with nonzero constant coefficient via recursion \eqref{Francis}, using length rather than degree simplifies relation \eqref{George} to
\begin{equation}\label{Henry}
\len f_n=2^n \len f_0.
\end{equation}

If $f=(f_0,f_1,\ldots,f_{\ell-1})$ and $g=(g_0,g_1,\ldots,g_{\ell-1})$ are two sequences of length $\ell$ and $s \in \Z$, then we define the {\it aperiodic crosscorrelation of $f$ and $g$ at shift $s$} to be 
\[
C_{f,g}(s) = \sum_{j \in \Z} f_{j+s} \conj{g_j},
\]
where we take $f_j=g_j=0$ whenever $j \not\in\{0,1,\ldots,\ell-1\}$.

Autocorrelation is crosscorrelation of a sequence with itself, so the {\it aperiodic autocorrelation of $f$ at shift $s$} is just $C_{f,f}(s)$.
If the terms of $f$ are complex numbers of unit magnitude, then $C_{f,f}(0)=\len f$.

Autocorrelation and crosscorrelation are studied extensively because of their importance in communications networks: see \cite{Scholtz-Welch,Sarwate-Pursley,Golomb-Gong,Karkkainen,Jedwab,Hoholdt,Schroeder} for some overviews.  It is desirable to have sequences whose autocorrelation values at all nonzero shifts are small in magnitude, and it is desirable to have pairs of sequences whose crosscorrelation values at all shifts are small in magnitude.

For a polynomial $a(z)=a_0+ a_1 z+ \cdots +a_{\ell-1} z^{\ell-1}\in \C[z]$, we define $\conj{a(z)}$ to be the Laurent polynomial $\conj{a_0} +\conj{a_1} z^{-1} + \cdots + \conj{a_{\ell-1}} z^{-(\ell-1)}$.  (We identify $\conj{z}$ with $z^{-1}$ because we are concerned with the properties of our polynomials on the complex unit circle.) 
If $f(z)$ and $g(z)$ are polynomials in $\C[z]$, then it is not hard to show that the values of the crosscorrelation between their associated sequences at all shifts are recorded in the following product of Laurent polynomials:
\[
f(z) \conj{g(z)} = \sum_{s \in \Z} C_{f,g}(s) z^s.
\]
The {\it crosscorrelation demerit factor} of $f$ and $g$ is defined to be
\[
\CDF(f,g) = \frac{\sum_{s \in \Z} |C_{f,g}(s)|^2}{|C_{f,f}(0)| \cdot |C_{g,g}(0)|}.
\]
Its reciprocal, the {\it crosscorrelation merit factor}, is defined as $\CMF(f,g) = 1/\CDF(f,g)$.  A low demerit factor (or equivalently, high merit factor) indicates a sequence pair whose crosscorrelation values are collectively low, hence desirable.
The {\it autocorrelation demerit factor of $f$} is much like the crosscorrelation demerit factor, but omits $|C_{f,f}(0)|^2$ in the numerator:
\begin{equation}\label{Linda}
\ADF(f) = \frac{\sum_{s \in \Z, s \not=0} |C_{f,f}(s)|^2}{|C_{f,f}(0)|^2} = \CDF(f,f)-1,
\end{equation}
and the {\it autocorrelation merit factor} is its reciprocal $\AMF(f) = 1/\ADF(f)$.  The autocorrelation merit factor defined here is Golay's original merit factor, introduced in \cite{Golay-1972}.

If $f(z) \in \C[z,z^{-1}]$ is a Laurent polynomial, and $p$ is a real number with $p\geq 1$, then we define the {\it $L^p$ norm of $f(z)$ on the complex unit circle} to be
\[
\norm{f}{p} = \left(\frac{1}{2\pi} \int_0^{2\pi} |f(e^{i \theta})|^p d\theta \right)^{1/p}.
\]
One can show (see \cite[Section V]{Katz}) that
\begin{equation}\label{David}
\CDF(f,g)=\cnrd{f}{g},
\end{equation}
and
\[
\ADF(f)=\CDF(f,f)-1 = \anrd{f}.
\]

Borwein and Mossinghoff \cite[Theorem 1 and Corollary 1]{Borwein-Mossinghoff} explicitly calculate the autocorrelation demerit factors for Rudin-Shapiro-like Littlewood polynomials and determine their asymptotic behavior.
\begin{theorem}[Borwein-Mossinghoff, 2000]\label{Robert}
If $f_0,f_1,f_2,\ldots$ is a sequence of Rudin-Shapiro-like polynomials generated from any Littlewood polynomial $f_0(z)$ via recursion \eqref{Francis}, then
\[
\lim_{n\to\infty} \ADF(f_n) =-1+ \frac{2}{3} \cdot \frac{\normff{f_0}+\normtt{f_0 \tf_0}}{\normtf{f_0}} \geq \frac{1}{3},
\]
where $\tf_0(z)$ is the polynomial $f_0(-z)$.
\end{theorem}
Borwein and Mossinghoff \cite[Section 3]{Borwein-Mossinghoff} go on to find examples where the limiting autocorrelation demerit factor reaches the lower bound of $1/3$, so that well-chosen families of Rudin-Shapiro-like polynomials can reach asymptotic autocorrelation merit factors as high as $3$.

We are interested in both autocorrelation and crosscorrelation merit factors.
It turns out that there are limits to how good one can simultaneously make autocorrelation and crosscorrelation performance.
Pursley and Sarwate \cite[eqs.~(3),(4)]{Pursley-Sarwate} proved a bound that relates autocorrelation and crosscorrelation demerit factors for binary sequences:
\begin{equation}\label{Nancy}
|\CDF(f,g)-1| \leq \sqrt{\ADF(f) \ADF(g)}.
\end{equation}
We define the {\it Pursley-Sarwate Criterion} of $f$ and $g$ to be
\begin{equation}\label{Paul}
\PSC(f,g)=\sqrt{\ADF(f)\ADF(g)} + \CDF(f,g),
\end{equation}
so that \eqref{Nancy} implies that
\begin{equation}\label{Hyeon}
\PSC(f,g) \geq 1.
\end{equation}
Since we want sequence pairs with low mutual crosscorrelation and where both sequences individually have low autocorrelation, we would like to find $f$ and $g$ with $\PSC(f,g)$ as close to $1$ as possible.

In pursuit of this goal, we found a formula for the asymptotic crosscorrelation demerit factor of a family of pairs of Rudin-Shapiro-like sequences.  This formula specializes to give information about autocorrelation that generalizes the results of Theorem \ref{Robert} to embrace polynomials with coefficients other than $-1$ and $1$.
\begin{theorem}\label{Vivian}
Let $f_0, g_0 \in \C[z]$ be polynomials of equal length having nonzero constant coefficients.
If $f_0,f_1,\ldots$ and $g_0,g_1,\ldots$ are sequences of Rudin-Shapiro-like polynomials generated from $f_0$ and $g_0$ via recursion \eqref{Francis}, then
\begin{align*}
\lim_{n \to \infty} \ADF(f_n) & = -1 + \frac{2}{3} \cdot \frac{\normff{f_0}+ \normtt{f_0 \tf_0}}{\normtf{f_0}} \geq \frac{1}{3}, \\
\lim_{n \to \infty} \ADF(g_n) & = -1 + \frac{2}{3} \cdot \frac{\normff{g_0}+ \normtt{g_0 \tg_0}}{\normtf{g_0}} \geq \frac{1}{3}, \\
\lim_{n \to \infty} \CDF(f_n,g_n) & = \frac{2 \normtt{f_0 g_0}+\normtt{f_0 \tg_0} + \Re  \int f_0 \tf_0 \conj{g_0 \tg_0}}{3 \normtt{f_0} \normtt{g_0}},
\end{align*}
where $\tf_0(z)$ and $\tg_0(z)$ are respectively the polynomials $f_0(-z)$ and $g_0(-z)$, and
\[
\int f_0 \tf_0 \conj{g_0 \tg_0} =
\frac{1}{2\pi} \int_{0}^{2\pi} f_0(e^{i \theta}) \tf_0(e^{i \theta}) \conj{g_0(e^{i\theta})  \tg_0(e^{i\theta})} d\theta.
\]
\end{theorem}
This theorem is proved in Corollary \ref{Samuel} in Section \ref{Lawrence}.  We then use the formula in Theorem \ref{Vivian} and computational searches (see Section \ref{Thomas}) to find families of pairs of Rudin-Shapiro-like Littlewood polynomials whose asymptotic Pursley-Sarwate Criterion is as low as $331/300=1.10333\ldots$, which is quite close to the absolute lower bound in \eqref{Hyeon}.    In contrast, the typical Pursley-Sarwate Criterion of randomly selected long binary sequences is about $2$, and high-performance sequence pairs constructed from finite field characters have been found with asymptotic Pursley-Sarwate Criterion of $7/6$ (see \cite[\S II.E, \S IV.D]{Katz} and \cite[eq.(6)]{Boothby-Katz}).  

The rest of this paper is organized as follows.  The goal of Section \ref{Lawrence} is to prove Theorem \ref{Vivian} above.  This is accomplished by finding recursive relations between various $L^p$ norms associated with our Rudin-Shapiro-like polynomials that arise from the original recursion \eqref{Francis}.

Section \ref{Simon} examines groups of symmetries that preserve the asymptotic correlation behavior when applied to our Rudin-Shapiro-like polynomials.  These are helpful in abbreviating computational searches (reported later in Section \ref{Thomas}) for polynomials with good autocorrelation and crosscorrelation performance.  We organize the good polynomials that we find into orbits modulo the action of our symmetry groups, which makes our reports shorter and more intelligible.

In Section \ref{Thomas}, we present some examples of Rudin-Shapiro-like Littlewood polynomials $f_0,f_1,\ldots$ and $g_0,g_1,\ldots$ such that $\lim_{n\to\infty} \PSC(f_n,g_n)$ is low, which implies simultaneously good autocorrelation and crosscorrelation performance.  This includes the families of polynomials with the exceptionally low asymptotic Pursley-Sarwate Criterion value reported above.

\section{Asymptotic Crosscorrelation Formula}\label{Lawrence}

In this section we prove Theorem \ref{Vivian}, the main theoretical result of this paper.
First we set down some notational conventions.

Throughout this paper, we let $\C[z,z^{-1}]$ denote the ring of Laurent polynomials with coefficients from $\C$.
Because we are working with polynomials on the complex unit circle, if $a(z)=\sum_{j \in \Z} a_j z^j$, then we use $\conj{a(z)}$ as a shorthand for $\sum_{j \in \Z} \conj{a_j} z^{-j}$, $\Re(a(z))$ as a shorthand for $\frac{1}{2} (a(z)+\conj{a(z)})$, and $|a(z)|^2$ as a shorthand for $a(z) \conj{a(z)}$.
We also use $\ta(z)$ as a shorthand for $a(-z)$.
Also recall from the Introduction that if $a(z)=a_0 + a_1 z + \cdots + a_d z^d$ is a polynomial of degree $d$ in $\C[z]$, then $a^\dag(z)$ denotes the conjugate reciprocal polynomial of $a(z)$, that is, $a^\dag(z)=\conj{a_d} + \conj{a_{d-1}} z + \cdots + \conj{a_0} z^d$.

We first note how the transformations $a\mapsto\ta$, $a\mapsto \conj{a}$, and $a\mapsto a^\dag$ relate to and interact with each other.
\begin{lemma}\label{Alice}
If $f(z) \in \C[z]$, then
\begin{enumerate}[(i).]
\item\label{William} $f^\dag(z)=z^{\deg f} \conj{f(z)}$,  
\item\label{Annie} $\widetilde{(f^\dag})(z)= (-z)^{\deg f} \conj{\tf(z)}=(-1)^{\deg f} \cdot (\tf)^\dag(z)$, and  
\item\label{Alexander} $\conj{f^\dag(z)}=z^{-\deg f} f(z)$.
\end{enumerate}
\end{lemma}
\begin{proof}
For part \eqref{William}, note that if $d=\deg f$ and $f(z)=f_0 + f_1 z + \cdots +f_d z^d$, then by the definition of the conjugate reciprocal, we have
\begin{align*}
f^\dag(z)
& =\conj{f_d} + \cdots + \conj{f_1} z^{d-1} + \conj{f_0} z^d \\
& = z^d(\conj{f_d} z^{-d} + \cdots + \conj{f_1} z^{-1} + \conj{f_0}) \\
& = z^d \conj{f(z)}.
\end{align*}
For part \eqref{Annie}, use part \eqref{William} to see that 
\begin{align*}
\widetilde{(f^\dag)}(z)
& = (-z)^{\deg f} \conj{f(-z)} \\
& = (-z)^{\deg f} \conj{\tf(z)},
\end{align*}
and then
\begin{align*}
(-z)^{\deg f} \conj{\tf(z)}
& = (-1)^{\deg f} \cdot z^{\deg \tf} \, \conj{\tf(z)} \\
& = (-1)^{\deg f} \cdot (\tf)^\dag(z),
\end{align*}
where we used part \eqref{William} again in the second equality.

For part \eqref{Alexander}, use part \eqref{William} to see that
\begin{align*}
\conj{f^\dag(z)}
& = \conj{z^{\deg f} \conj{f(z)}} \\
& = z^{-\deg f} f(z).\qedhere
\end{align*}
\end{proof}

If $a(z) \in \C[z,z^{-1}]$, say $a(z)=\sum_{j \in \Z} a_j z^j$, then we use $\int a(z)$ as a shorthand for $a_0$.
This is because if we actually perform integration on the complex unit circle, we obtain $a_0$:
\[
\frac{1}{2\pi} \int_0^{2\pi} a(e^{i \theta}) d\theta  = a_0.
\]
In particular note that
\[
\normtt{a(z)} = \int a(z) \conj{a(z)}=\sum_{j \in \Z} |a_j|^2.
\]

It will be important to know that replacing $a$ with $\ta$ changes neither integrals nor norms.
\begin{lemma}\label{Leonard} \quad
\begin{enumerate}[(i).]
\item\label{Martha} For any $f(z) \in \C[z,z^{-1}]$, we have $\int \tf(z) = \int f(z)$.
\item\label{Norbert} For any $f(z) \in \C[z,z^{-1}]$ and any $p \in \R$ with $p \geq 1$, we have $\norm{\tf(z)}{p}=\norm{f(z)}{p}$.
\end{enumerate}
\end{lemma}
\begin{proof}
If $g(z)$ is any function that is integrable on the complex unit circle, then
\begin{align*}
\frac{1}{2\pi} \int_0^{2\pi} g(-e^{i\theta}) d\theta
& = \frac{1}{2\pi} \int_0^{\pi} g(e^{i(\theta+\pi)}) d\theta + \frac{1}{2\pi} \int_{\pi}^{2 \pi} g(e^{i(\theta-\pi)}) d\theta \\
& = \frac{1}{2\pi} \int_\pi^{2 \pi} g(e^{i\eta}) d\eta + \frac{1}{2\pi} \int_{0}^{\pi} g(e^{i\eta}) d\eta \\
& = \frac{1}{2\pi} \int_0^{2\pi} g(e^{i\theta}) d\theta.
\end{align*}
This proves parts \eqref{Martha} and \eqref{Norbert}, where we set $g(z)=f(z)$ or $g(z)=|f(z)|^p$, respectively.
\end{proof}

If we have two sequences $f_0,f_1,\ldots$ and $g_0,g_1,\ldots$ of Rudin-Shapiro-like polynomials constructed via recursion \eqref{Francis}, then the following lemma tells us how $\normtt{f_{n+1} g_{n+1}}$ is related to $\normtt{f_n g_n}$.
In view of \eqref{David}, this is telling us how $\CDF(f_n,g_n)$ changes in one step of the recursion.
\begin{lemma}\label{Albert}
Suppose that $n$ is a nonnegative integer, and that $f_n(z),g_n(z) \in \C[z]$ are polynomials of length $\ell$, and let $f_{n+1}(z)=f_n(z)+\sigma_n z^\ell f_n^\dag(-z)$ and $g_{n+1}(z)=g_n(z)+\tau_n z^\ell g_n^\dag(-z)$, where $\sigma_n,\tau_n \in \{-1,1\}$.   Then $\normtt{f_{n+1}}=2\normtt{f_n}$, $\normtt{g_{n+1}}=2\normtt{g_n}$, and if we define
\begin{align*}
u_j & =\normtt{f_j g_j}, \\
v_j & =\normtt{f_j \tg_j}, \\
w_j & = \Re \int f_j \tf_j \conj{g_j \tg_j} 
\end{align*}
for $j\in\{n,n+1\}$, then
\begin{enumerate}[(i).]
\item $u_{n+1}=2 u_n + 2 v_n + 2 \sigma_n\tau_n w_n$,
\item $v_{n+1}=2 u_n + 2 v_n - 2 \sigma_n\tau_n w_n$, and
\item $w_{n+1}=2 \sigma_n\tau_n u_n - 2 \sigma_n\tau_n v_n + 2 w_n$.
\end{enumerate}
\end{lemma}
\begin{proof}
First observe that Lemma \ref{Alice}\eqref{Annie} shows that
\begin{align}
f_{n+1}(z)& =f_n(z)+\sigma_n \modeg z^{2\ell-1} \conj{\tf_n(z)} \label{Harold} \\
g_{n+1}(z)& =g_n(z)+\tau_n \modeg z^{2\ell-1} \conj{\tg_n(z)}, \nonumber
\end{align}
and so
\begin{align*}
\normtt{f_{n+1}}
& = \int \left|f_n+\sigma_n \modeg z^{2\ell-1} \conj{\tf_n}\right|^2 \\
& = \int |f_n|^2+|\tf_n|^2 + 2 \Re(\sigma_n \modeg z^{1-2\ell} f_n \tf_n) \\
& = 2 \normtt{f_n} + 2\sigma_n \modeg \Re \int z^{1-2\ell} f_n \tf_n
\end{align*}
where the last equality uses Lemma \ref{Leonard}\eqref{Norbert}.  Now observe that $z^{1-2\ell} f_n \tf_n$ is a Laurent polynomial whose terms all have negative powers of $z$ (because $f_n$ is a polynomial of degree $\ell-1$), and so the last integral is zero.  Thus we obtain the desired result that $\normtt{f_{n+1}}=2\normtt{f_n}$.  If one replaces every instance of $f$ with $g$ in the above, one obtains a proof that $\normtt{g_{n+1}}=2\normtt{g_n}$.

Now we prove the recursions involving $u_n$, $v_n$, and $w_n$.  First of all,
\begin{align*}
u_{n+1}
& = \normtt{f_{n+1} g_{n+1}} \\
& = \frac{1}{2} \normtt{f_{n+1} g_{n+1}} + \frac{1}{2} \normtt{\tf_{n+1} \tg_{n+1}} \\
& = \frac{1}{2} \bnormtt{\left(f_n+\sigma_n \modeg \zs \conj{\tf_n}\right)\left(g_n+\tau_n \modeg \zs \conj{\tg_n}\right)} \\
& \qquad+ \frac{1}{2} \bnormtt{\left(\tf_n-\sigma_n\modeg \zs \conj{f_n}\right)\left(\tg_n-\tau_n\modeg \zs \conj{g_n}\right)} \\
& = 2 \normtt{f_n g_n} + 2 \normtt{f_n \tg_n} + 2\sigma_n \tau_n \Re\int f_n \tf_n \conj{g_n \tg_n} \\
& \qquad + 2\sigma_n \tau_n \Re\int z^{2-4\ell} f_n \tf_n g_n \tg_n \\
& = 2 u_n + 2 v_n +  2 \sigma_n\tau_n w_n + 2 \sigma_n\tau_n \Re \int \idzs f_n \tf_n g_n \tg_n,
\end{align*}
where the first equality is the definition of $u_{n+1}$, the second equality uses Lemma \ref{Leonard}\eqref{Norbert}, the third equality uses \eqref{Harold}, and the fourth equality uses technical Lemma \ref{Eustace}, which appears at the end of this section.
Then note that $\idzs f_n \tf_n g_n \tg_n$ is a Laurent polynomial whose terms all have negative powers of $z$ (because $f_n$ and $g_n$ are polynomials of degree $\ell-1$), so that the last integral in our chain of equalities is zero, giving us the desired result.

We also have
\begin{align*}
v_{n+1}
& = \normtt{f_{n+1} \tg_{n+1}} \\
& = \frac{1}{2} \normtt{f_{n+1} \tg_{n+1}} + \frac{1}{2} \normtt{\tf_{n+1} g_{n+1}} \\
& = \frac{1}{2} \bnormtt{\left(f_n+\sigma_n \modeg \zs \conj{\tf_n}\right)\left(\tg_n-\tau_n\modeg \zs \conj{g_n}\right)} \\
& \qquad+ \frac{1}{2} \bnormtt{\left(\tf_n-\sigma_n\modeg \zs \conj{f_n}\right)\left(g_n+\tau_n \modeg \zs \conj{\tg_n}\right)} \\
& = 2 \normtt{f_n \tg_n} + 2 \normtt{f_n g_n} + 2\sigma_n (-\tau_n) \Re\int f_n \tf_n \conj{\tg_n g_n} \\
& \qquad + 2\sigma_n (-\tau_n) \Re\int z^{2-4\ell} f_n \tf_n \tg_n g_n \\
& = 2 u_n + 2 v_n - 2 \sigma_n\tau_n w_n - 2 \sigma_n\tau_n \Re \int \idzs f_n \tf_n g_n \tg_n,
\end{align*}
where the first equality is the definition of $v_{n+1}$, the second equality uses Lemma \ref{Leonard}\eqref{Norbert}, the third equality uses \eqref{Harold}, and the fourth equality uses technical Lemma \ref{Eustace}, which appears at the end of this section.
Then note that $\idzs f_n \tf_n g_n \tg_n$ is a Laurent polynomial whose terms all have negative powers of $z$ (because $f_n$ and $g_n$ are polynomials of degree $\ell-1$), so that the last integral in our chain of equalities is zero, giving us the desired result.

Finally, we use similar arguments to obtain
\begin{align*}
w_{n+1}
& =\Re \int f_{n+1} \tf_{n+1} \conj{g_{n+1} \tg_{n+1}} \\
& = \Re \int \left(f_n+\sigma_n \modeg \zs \conj{\tf_n}\right) \left(\tf_n-\sigma_n \modeg \zs \conj{f_n}\right)  \\ 
& \qquad\qquad \times \left(\conj{g_n}+\tau_n \modeg \izs \tg_n\right) \left(\conj{\tg_n}-\tau_n \modeg \izs g_n\right)  \\
& = \Re \int \left[f_n \tf_n - \dzs \conj{f_n \tf_n} +\sigma_n \modeg \zs \left(|\tf_n|^2 -|f_n|^2\right)\right] \\
& \qquad\qquad \times \left[\conj{g_n \tg_n} - \idzs g_n \tg_n +\tau_n \modeg \izs \left(|\tg_n|^2 -|g_n|^2\right)\right]  \\
& = I_1 + \sigma_n\tau_n I_2 + \tau_n \modeg I_3 + \sigma_n \modeg I_4,
\end{align*}
where
\begin{align*}
I_1 & = \Re \int \left(f_n \tf_n - \dzs \conj{f_n \tf_n}\right) \left(\conj{g_n \tg_n} - \idzs g_n \tg_n\right) \\
I_2 & = \Re \int \left(|\tf_n|^2 -|f_n|^2\right) \left(|\tg_n|^2- |g_n|^2 \right)\\
I_3 & = \Re \int \left(f_n \tf_n - \dzs \conj{f_n \tf_n}\right) \izs \left(|\tg_n|^2- |g_n|^2 \right) \\
I_4 & = \Re \int \zs \left(|\tf_n|^2 -|f_n|^2\right) \left(\conj{g_n \tg_n} - \idzs g_n \tg_n\right).
\end{align*}
Now we compute each of these four integrals.
\begin{align*}
I_1
& = \Re \int \left(f_n \tf_n - \dzs \conj{f_n \tf_n}\right) \left(\conj{g_n \tg_n} - \idzs g_n \tg_n\right) \\
& = 2 \Re \int f_n \tf_n \conj{g_n \tg_n} - 2 \Re \int \idzs f_n \tf_n g_n \tg_n,
\end{align*}
and note that $\idzs f_n \tf_n g_n \tg_n$ is a Laurent polynomial whose terms all have negative powers of $z$ (because $f_n$ and $g_n$ are of length $\ell-1$), so that the last integral is zero, and thus
\[
I_1 = 2 w_n.
\]
Then
\begin{align*}
I_2
& = \normtt{\tf_n \tg_n}+\normtt{f_n g_n} - \normtt{\tf_n g_n} -\normtt{f_n \tg_n}\\
& = 2 u_n - 2 v_n,
\end{align*}
where the second equality is due to Lemma \ref{Leonard}\eqref{Norbert}.

Let us examine the integrand in the definition of $I_3$, which is
\[
h=\left(\izs f_n \tf_n - \zs \conj{f_n \tf_n}\right) \left(|\tg_n|^2-|g_n|^2\right).
\]
If one conjugates this, one obtains
\[
\left(\zs \conj{f_n \tf_n} - \izs f_n \tf_n\right) \left(|\tg_n|^2-|g_n|^2\right),
\]
which is just $-h$.  So $h$ has purely imaginary values on the unit circle, and thus $I_3=\Re \int h=0$.
The same argument shows that $I_4=0$, and so, putting all our results together, we have
\begin{align*}
w_{n+1}
& = I_1 + \sigma_n\tau_n I_2 + \tau_n \modeg I_3 + \sigma_n \modeg I_4 \\
& = 2 w_n + \sigma_n\tau_n(2 u_n-2 v_n) + 0 + 0.\qedhere
\end{align*}
\end{proof}
The above lemma allows us to compute crosscorrelation demerit factors for pairs of Rudin-Shapiro-like polynomials constructed via recursion \eqref{Francis}.
\begin{theorem}\label{Sally}
Let $f_0, g_0 \in \C[z]$ be polynomials of equal length having nonzero constant coefficients.
Let $\sigma_0,\sigma_1,\ldots$ be a sequence of values from $\{-1,1\}$, and suppose that $f_n(z)$ and $g_n(z)$ are defined recursively for all $n \in \N$ by
\begin{align*}
f_{n+1}(z) & =f_n(z)+\sigma_n z^{\len f_n} f_n^\dag(-z) \\
g_{n+1}(z) & =g_n(z)+\sigma_n z^{\len g_n} g_n^\dag(-z).
\end{align*}
Then
\begin{multline*}
\cnrd{f_n}{g_n} = \frac{2 \normtt{f_0 g_0}+\normtt{f_0 \tg_0} +  \Re \int f_0 \tf_0 \conj{g_0 \tg_0}}{3 \normtt{f_0} \normtt{g_0}} \\
+\left(-\frac{1}{2}\right)^n \frac{\normtt{f_0 g_0} -\normtt{f_0 \tg_0} - \Re \int f_0 \tf_0 \conj{g_0 \tg_0}}{3 \normtt{f_0} \normtt{g_0}}.
\end{multline*}
\end{theorem}
\begin{proof}
Since $f_0$ and $g_0$ have nonzero constant coefficients and are of the same length, induction shows that for every $n$, the polynomials $f_n$ and $g_n$ have nonzero constant coefficients and both are of length $2^n \len f_0=2^n \len g_0$, as observed in \eqref{Henry} in the Introduction.
Thus we may apply Lemma \ref{Albert} repeatedly to the pairs $(f_n,g_n)$ for every $n$.
  
Let $u_j$, $v_j$, and $w_j$ be as defined in Lemma \ref{Albert}. We want to calculate $u_n$, and the lemma says that
\[
\begin{pmatrix} u_{n+1} \\ v_{n+1} \\ w_{n+1} \end{pmatrix} = A \begin{pmatrix} u_n \\ v_n \\ w_n \end{pmatrix},
\]
where
\[
A=\begin{pmatrix} 2 & 2 & 2 \\ 2 & 2 & -2 \\ 2 & -2 & 2 \end{pmatrix}.
\]
Now $A=B \Lambda B^{-1}$, where
\[
B=\begin{pmatrix} -1 & 1 & 1 \\ 1 & 1 & 0 \\ 1 & 0 & 1 \end{pmatrix} \text{\qquad and \qquad} \Lambda = \begin{pmatrix} -2 & 0 & 0 \\ 0 & 4 & 0 \\ 0 & 0 & 4 \end{pmatrix}.
\]
So
\begin{align*}
\begin{pmatrix} u_n \\ v_n \\ w_n \end{pmatrix}
& = A^n \begin{pmatrix} u_0 \\ v_0 \\ w_0 \end{pmatrix} \\
& = B \Lambda^n B^{-1}  \begin{pmatrix} u_0 \\ v_0 \\ w_0 \end{pmatrix},
\end{align*}
and so
\[
\normtt{f_n g_n} = u_n = \frac{4^n(2 u_0+v_0+w_0) + (-2)^n (u_0-v_0-w_0)}{3}.
\]

Repeated use of Lemma \ref{Albert} also shows that $\normtt{f_n}=2^n \normtt{f_0}$ and $\normtt{g_n}=2^n\normtt{g_0}$, and these norms are nonzero since $f_0$ and $g_0$ are nonzero, so that
\[
\cnrd{f_n}{g_n}= \frac{(2 u_0+v_0+w_0)+(-1/2)^n (u_0-v_0-w_0)}{3 \normtt{f_0} \normtt{g_0}},
\]
and when one substitutes the values of $u_0$, $v_0$, and $w_0$ as defined in Lemma \ref{Albert}, then one obtains the desired result.
\end{proof}
If $f_0=g_0$, we are considering autocorrelation.
When we specialize to this case and also specialize to the case where $f_0$ is a Littlewood polynomial, then we recover the results of Borwein and Mossinghoff \cite[Theorem 1 and Corollary 1]{Borwein-Mossinghoff}.
\begin{corollary}[Borwein-Mossinghoff (2000)]\label{Barbara}
Suppose that $f_0(z)$ is a Littlewood polynomial and that $f_n(z)$ is defined recursively for all $n \in \N$ by
\[
f_{n+1}(z) = f_n(z)+z^{\len f_n} f_n^\dag(-z).
\]
Then
\[
\frac{\normff{f_n}}{\normtf{f_n}} = \frac{2}{3} \cdot \frac{\normff{f_0}+\normtt{f_0 \tf_0}}{\normtf{f_0}} +\left(-\frac{1}{2}\right)^n \cdot \frac{1}{3} \cdot \frac{\normff{f_0} -2 \normtt{f_0 \tf_0}}{\normtf{f_0}}.
\]
\end{corollary}
\begin{remark}
Borwein and Mossinghoff use $\normtt{f_0 \widetilde{f_0^*}}$ instead of our $\normtt{f_0 \tf_0}$, but it is not hard to see that these are equal because $\widetilde{f^*}(z)=\widetilde{f^\dag}(z)=(-z)^{\deg f} \conj{\tf(z)}$ so that $|\widetilde{f^*}|^2=|\widetilde{f}|^2$ for any Littlewood polynomial $f$ (see Lemma \ref{Alice}\eqref{Annie}).
\end{remark}
Our results now allow us to compute limiting autocorrelation and crosscorrelation demerit factors.
The following corollary contains all the results that we presented in Theorem \ref{Vivian} in the Introduction.
\begin{corollary}\label{Samuel}
Let $f_0, g_0 \in \C[z]$ be polynomials of equal length having nonzero constant coefficients.
Let $\sigma_0,\sigma_1,\ldots$ be a sequence of values from $\{-1,1\}$, and suppose that $f_n(z)$ and $g_n(z)$ are defined recursively for all $n \in \N$ by
\begin{align*}
f_{n+1}(z) & =f_n(z)+\sigma_n z^{\len f_n} f_n^\dag(-z) \\
g_{n+1}(z) & =g_n(z)+\sigma_n z^{\len g_n} g_n^\dag(-z).
\end{align*}
Then
\begin{align*}
\lim_{n \to \infty} \ADF(f_n) & = -1 + \frac{2}{3} \cdot \frac{\normff{f_0}+ \normtt{f_0 \tf_0}}{ \normtf{f_0}} \geq \frac{1}{3}, \\
\lim_{n \to \infty} \ADF(g_n) & = -1 + \frac{2}{3} \cdot \frac{\normff{g_0}+ \normtt{g_0 \tg_0}}{ \normtf{g_0}} \geq \frac{1}{3}, \\
\lim_{n \to \infty} \CDF(f_n,g_n) & = \frac{2 \normtt{f_0 g_0}+\normtt{f_0 \tg_0} +  \Re \int f_0 \tf_0 \conj{g_0 \tg_0}}{3 \normtt{f_0} \normtt{g_0}},
\end{align*}
so that
\begin{align*}
\lim_{n \to\infty} & \PSC(f_n,g_n) =  \frac{2 \normtt{f_0 g_0}+\normtt{f_0 \tg_0} +  \Re \int f_0 \tf_0 \conj{g_0 \tg_0}}{3 \normtt{f_0} \normtt{g_0}} \\ & + \frac{\sqrt{\left(2\normff{f_0}+2 \normtt{f_0 \tf_0}-3\normtf{f_0}\right) \left(2\normff{g_0}+2\normtt{g_0 \tg_0}-3\normtf{g_0}\right)}}{3 \normtt{f_0} \normtt{g_0}}.
\end{align*}
\end{corollary}
\begin{proof}
The limiting crosscorrelation demerit factor is clear from Theorem \ref{Sally} since the ratio of norms calculated there is the crosscorrelation demerit factor by \eqref{David}.  For the limiting autocorrelation demerit factors, one again uses Theorem \ref{Sally}, but now one sets $f_n=g_n$ for all $n$ in that theorem, and combines the result thus obtained with the fact from \eqref{Linda} that $\ADF(f_n)=\CDF(f_n,f_n)-1$ along with the observations that $\normtt{f_0 f_0}=\normff{f_0}$ and $\Re \int f_0 \tf_0 \conj{f_0 \tf_0}=\normtt{f_0 \tf_0}$.  The limiting Pursley-Sarwate Criterion follows immediately from the definition in \eqref{Paul} and the limits on the autocorrelation and crosscorrelation demerit factors.

To obtain the lower bounds on the limiting autocorrelation demerit factors, one notes that for any $f(z) \in \C[z]$, we have
\begin{align*}
\normff{f}+ \normtt{f \tf}
& = \frac{1}{2} \normff{f} + \frac{1}{2} \normff{\tf} + \normtt{f \tf} \\
& = \frac{1}{2} \normtt{(|f(z)|^2 + |\tf(z)|^2)}  \\
& \geq \frac{1}{2} \normot{(|f(z)|^2 + |\tf(z)|^2)}  \\
& = \frac{1}{2} \left(\normtt{f}+\normtt{\tf}\right)^2 \\
& = \frac{1}{2} \left(2 \normtt{f}\right)^2 \\
& = 2 \normtf{f},
\end{align*}
where we use Lemma \ref{Leonard}\eqref{Norbert} in the first and penultimate equalities, and the inequality uses the fact that the $L^2$ norm is always at least as large as the $L^1$ norm (by Jensen's inequality) because we are working on a space of measure $1$.
Therefore, $\left(\normff{f}+ \normtt{f \tf}\right)/\normtf{f} \geq 2$, which proves our lower bounds on the limiting values autocorrelation demerit factors.
\end{proof}
Notice that although Corollary \ref{Samuel} has lower bounds on limiting autocorrelation demerit factors, no lower bound for the crosscorrelation demerit factor is given.  This is because the only lower bound that could have been given is the trivial lower bound of $0$, for the following example shows that one can obtain pairs of stems whose limiting crosscorrelation demerit factors are arbitrarily close to $0$.
\begin{proposition}\label{Elaine}
Let $k$ be a positive integer, let $f_0(z)=(1-z^{4 k})/(1-z)$, and $g_0(z)=(1-z)(1-z^2)(1-z^{4 k})/(1-z^4)$.
If $f_n(z)$ and $g_n(z)$ are defined recursively for all $n \in \N$ by
\begin{align*}
f_{n+1}(z) & =f_n(z)+ z^{\len f_n} f_n^\dag(-z) \\
g_{n+1}(z) & =g_n(z)+ z^{\len g_n} g_n^\dag(-z),
\end{align*}
then
\[
\lim_{n \to \infty} \CDF(f_n,g_n) = \frac{1}{3 k}.
\]
\end{proposition}
We shall prove this after some brief comments.
\begin{remark}\label{Raphael}
Note that $f_0$ in Proposition \ref{Elaine} is a Littlewood polynomial representing a sequence of length $4 k$ whose terms are all $1$.  And $g_0$ is Littlewood polynomial representing a sequence of length $4 k$ consisting of $k$ repetitions of the smaller sequence $(1,-1,-1,1)$.

Proposition \ref{Elaine} shows that we can find a pair of stems whose limiting crosscorrelation demerit factor is as close to $0$ as we like simply by choosing a sufficiently high value of $k$ when we define $f_0$ and $g_0$.

In view of the Pursley-Sarwate bound \eqref{Nancy}, we know that the autocorrelation performance for such stems cannot be exceptionally good, and in fact, one can use Corollary \ref{Samuel} to calculate the limiting autocorrelation demerit factors for the stems from seeds $f_0$ and $g_0$ described in Proposition \ref{Elaine}.  It is easy to calculate the sums of squares of the autocorrelation values for $f_0$ and for $g_0$ and also to compute the sums of squares of the crosscorrelation values for $f_0$ with $\tf_0$ and for $g_0$ with $\tg_0$ to show that $\normff{f_0}=4 k(32 k^2+1)/3$, $\normff{g_0}=4 k(16 k^2+5)/3$, $\normtt{f_0 \tf_0}=4 k$, and $\normtt{g_0 \tg_0}=4 k(16 k^2-1)/3$.  Since $f_0$ and $g_0$ are Littlewood polynomials of length $4 k$, we have $\normtt{f_0}=\normtt{g_0}=4 k$, so that if $f_0,f_1,\ldots$ and $g_0,g_1,\ldots$ are the stems obtained from seeds $f_0$ and $g_0$ by recursion \eqref{Francis}, then Corollary \ref{Samuel} tells us that $\lim_{n\to \infty} \ADF(f_n)=\lim_{n\to\infty} \ADF(g_n)= (16 k^2 - 9 k+2)/(9 k)$, which is strictly increasing from a value of $1$ (when $k=1$) to $\infty$ in the limit as $k\to \infty$.
\end{remark}
\begin{proof}[Proof of Proposition \ref{Elaine}]
We shall use Corollary \ref{Samuel} to calculate the limiting crosscorrelation demerit factor.  To that end, we calculate
\begin{align*}
f_0(z) g_0(z) 
& = \frac{(1-z^{4 k})^2 (1-z)(1-z^2)}{(1-z)(1-z^4)} \\
& = \left(\frac{1-z^{4 k}}{1-z^4}\right) (1-z^2) (1-z^{4 k}) \\
& = (1-z^2+z^4-z^6+\cdots+z^{4 k-4}-z^{4 k-2}) (1-z^{4 k}),
\end{align*}
which is a polynomial with $4 k$ nonzero coefficients, every of one of which is either $1$ or $-1$, so then $\normtt{f_0 g_0}=4 k$.

And then we calculate
\begin{align*}
f_0(z) \tg_0(z)
& = \frac{(1-z^{4 k})^2 (1+z)(1-z^2)}{(1-z)(1-z^4)} \\
& = \left(\frac{1-z^{4 k}}{1-z^4}\right) (1+z)^2 (1-z^{4 k}) \\
& = (1+2 z+ z^2+\cdots+z^{4 k-4}+2 z^{4 k-3}+z^{4 k-2}) (1-z^{4 k}),
\end{align*}
which is a polynomial with $4 k$ coefficients of magnitude $1$ and $2 k$ coefficients of magnitude $2$, so then $\normtt{f_0 \tg_0}=4 k+4 \cdot 2 k= 12 k$.

And then we calculate
\begin{align*}
f_0(z) \tf_0(z) \conj{g_0(z) \tg_0(z)}
 & = \frac{(1-z^{4 k})^2 (1-z^{-4 k})^2(1-z^{-2})^2(1-z^{-1})(1+z^{-1})}{(1-z)(1+z)(1-z^{-4})^2} \\
& = \frac{(1-z^{4 k})^2 (1-z^{-4 k})^2(1-z^{-2})^2(-z^{-1})(z^{-1})}{(1-z^{-4})^2} \\
& = -\frac{(1-z^{4 k})^2 (1-z^{-4 k})^2(1-z^{-2})^2 z^2}{(1-z^{-4})^2 z^4} \\
& = -\frac{(1-z^{4 k})^2 (1-z^{-4 k})^2(1-z^2) (1-z^{-2})}{(1-z^4)(1-z^{-4})}.
\end{align*}
Thus $\int \! f_0(z) \tf_0(z) \conj{g_0(z) \tg_0(z)} = - \normtt{(1-z^{4 k})^2(1-z^2)/(1-z^4)}$, and we have already calculated the norm to be $4 k$, and so $\int f_0(z) \tf_0(z) \conj{g_0(z) \tg_0(z)}=-4 k$.

Finally, $\normtt{f_0}=\normtt{g_0}=4 k$ because $f_0$ and $g_0$ are Littlewood polynomials of length $4 k$.
Now Corollary \ref{Samuel} says that if $f_n(z)$ and $g_n(z)$ are defined recursively for all $n \in \N$ by
\begin{align*}
f_{n+1}(z) & =f_n(z)+ z^{\len f_n} f_n^\dag(-z) \\
g_{n+1}(z) & =g_n(z)+ z^{\len g_n} g_n^\dag(-z),
\end{align*}
then
\begin{align*}
\lim_{n \to \infty} \CDF(f_n,g_n) & = \frac{2 \normtt{f_0 g_0}+\normtt{f_0 \tg_0} +  \Re \int f_0 \tf_0 \conj{g_0 \tg_0}}{3 \normtt{f_0} \normtt{g_0}} \\
& = \frac{2 \cdot 4 k + 12 k + \Re(-4 k)}{3 (4 k)^2} \\
& = \frac{1}{3 k}.\qedhere
\end{align*}
\end{proof}

We close this section with the technical lemma used in the proof of Lemma \ref{Albert} above.
\begin{lemma}\label{Eustace}
If $a(z), b(z) \in \C[z,z^{-1}]$, $k \in \Z$, and $\sigma,\tau \in \{-1,1\}$, and if
\[
I=\frac{1}{2} \bnormtt{\left(a+\sigma z^k \conj{\ta}\right)\left(b+\tau z^k \conj{\tb}\right)}+ \frac{1}{2} \bnormtt{\left(\ta-\sigma z^k \conj{a}\right)\left(\tb-\tau z^k \conj{b}\right)},\]
then
\[
I= 2\normtt{a b} + 2 \normtt{a \tb} + 2 \sigma\tau \Re \int a \ta \conj{b \tb} + 2\sigma \tau \Re \int z^{-2 k} a \ta b \tb.
\]
\end{lemma}
\begin{proof}
Note that
\begin{align*}
I & = \frac{1}{2} \int \left(|a|^2 + |z^k \conj{\ta}|^2  + 2\sigma \Re\left(a \conj{z}^k \ta\right) \right) \left(|b|^2 + |z^k \conj{\tb}|^2  + 2 \tau\Re\left(b \conj{z}^k \tb\right) \right)  \\
& \quad + \frac{1}{2} \int \left(|\ta|^2 + |z^k \conj{a}|^2  - 2 \sigma \Re\left(\ta \conj{z}^k a\right) \right) \left(|\tb|^2 + |z^k\conj{b}|^2  - 2 \tau \Re\left(\tb \conj{z}^k b\right) \right),
\end{align*}
and since we are integrating on the complex unit circle, we may omit terms of the form $|z^k|$ and replace $\conj{z}^k$ with $z^{-k}$ to obtain
\begin{align*}
I & = \frac{1}{2} \int \left(|a|^2 + |\ta|^2  + 2 \sigma \Re\left(z^{-k} a \ta\right) \right) \left(|b|^2 + |\tb|^2  + 2 \tau \Re\left(z^{-k} b \tb\right) \right)  \\
& \quad +  \frac{1}{2} \int \left(|\ta|^2 + |a|^2  - 2 \sigma \Re\left(z^{-k} \ta a\right) \right) \left(|\tb|^2 + |b|^2  - 2 \tau \Re\left(z^{-k} \tb b\right) \right),
\end{align*}
from which one obtains
\begin{align*}
I
& =\int \left(|a|^2 + |\ta|^2\right)  \left(|b|^2 + |\tb|^2\right)+ 4 \sigma\tau \int \Re\left(z^{-k} a \ta\right) \Re\left(z^{-k} b \tb\right) \\
& =\normtt{a b} + \normtt{\ta \tb} + \normtt{a \tb } +\normtt{\ta b} + 4 \sigma\tau \int \Re\left(z^{-k} a \ta\right) \Re\left(z^{-k} b \tb\right) \\
& =2\normtt{a b} + 2 \normtt{a \tb } +4 \sigma\tau \int \Re\left(z^{-k} a \ta\right) \Re\left(z^{-k} b \tb\right) \\
& =2\normtt{a b} + 2 \normtt{a \tb } +2\sigma\tau \int \left[\Re\left(z^{-k} a \ta z^{-k} b \tb\right)+\Re\left(z^{-k} a \ta \conj{z^{-k} b \tb}\right)\right],
\end{align*}
where the third equality uses Lemma \ref{Leonard}\eqref{Norbert} and the fourth equality uses the observation that $2 \Re(u)\Re(v)=\Re(u v)+\Re(u\conj{v})$.  The desired result now readily follows.
\end{proof}

\section{Symmetry Groups}\label{Simon}

The expressions in Theorem \ref{Vivian} for the limiting autocorrelation and crosscorrelation demerit factors are invariant under certain symmetries.  This helps abbreviate computational searches for sequences and sequence pairs with optimum performance.
These symmetries are based on negation of polynomials, replacement of $z$ by $-z$ in polynomials, and transformation of polynomials to their conjugate reciprocals.
One should recall the notational conventions $\tf(z)$ and $\conj{f(z)}$ for $f(z) \in \C[z,z^{-1}]$ and the definition of the conjugate reciprocal $f^\dag(z)$ for $f(z) \in \C[z]$ from the second paragraph of Section \ref{Lawrence}.
One should note that $\widetilde{f g}(z) = \tf(z) \tg(z)$ and $\conj{f(z) g(z)}=\conj{f(z)} \cdot \conj{g(z)}$ for every $f(z), g(z) \in \C[z,z^{-1}]$ and $(f(z)g(z))^\dag=f^\dag(z) g^\dag(z)$ for every $f(z), g(z) \in \C[z]$.  (The third relation follows easily from Lemma \ref{Alice}\eqref{William} and the second relation.)  We shall also need the following observation.
\begin{lemma}\label{Orestes} For any $f(z) \in \C[z]$ and any $p \in \R$ with $p \geq 1$, we have $\norm{f^\dag(z)}{p}=\norm{f(z)}{p}$.
\end{lemma}
\begin{proof}
By Lemma \ref{Alice}\eqref{William}, we have
\begin{align*}
\frac{1}{2\pi} \int_0^{2\pi} |f^\dag(e^{i\theta})|^p d\theta 
& = \frac{1}{2\pi} \int_0^{2\pi} \left|(e^{i\theta})^{\deg f} \conj{f(e^{i\theta})}\right|^p d\theta \\
& = \frac{1}{2\pi} \int_0^{2\pi} |f(e^{i\theta})|^p d\theta. \qedhere
\end{align*}
\end{proof}

Now we introduce a group of symmetries that, when applied to a polynomial $f_0(z) \in \C[z]$ with nonzero constant coefficient, will preserve the autocorrelation properties of the stem $f_0,f_1,\ldots$ of Rudin-Shapiro-like polynomials obtained from seed $f_0$ via recursion \eqref{Francis}.
First we describe our group and how it affects certain norms, and then we show its effect on autocorrelation as a corollary.
\begin{proposition}\label{Penelope}
Let $\ell$ be a nonnegative integer, and let $\apols$ be the set of all polynomials of length $\ell$ in $\C[z]$ that have nonzero constant coefficient.
We define three maps from $\apols$ to itself:
\begin{align*}
n(f) & = -f \\
h(f) & = \tf \\
r(f) & = f^\dag.
\end{align*}
These maps generate a group $\agr=\ggen{n,h,r}$ of permutations of $\apols$.
\begin{enumerate}[(i).]
\item If $\ell=1$, then $G_1$ is the internal direct product of the two cyclic groups $\ggen{n}$ and $\ggen{r}$, each of order $2$.
\item If $\ell$ is odd with $\ell > 1$, then $\agr$ is the internal direct product of the three cyclic groups $\ggen{n}$, $\ggen{h}$, and $\ggen{r}$, each of order $2$.
\item If $\ell$ is even, then $\agr$ is isomorphic to the dihedral group of order $8$ generated by $r h$ and $h$, where $r h$ is of order $4$, $h$ is of order $2$, and $h(r h) h^{-1}=(r h)^{-1}$.
\end{enumerate}
For any $t \in \agr$, any $f \in \apols$, and any $p \geq 1$, we have
\begin{align*}
\norm{t(f)}{p} & =\norm{f}{p} \\
\norm{t(f)\widetilde{t(f)}}{p} & =\norm{f\tf}{p}.
\end{align*}
\end{proposition}
\begin{proof}
It is clear that each of $n$, $h$, and $r$ is an involution on $\apols$ (except that $h$ is the identity element when $\ell=1$), so these maps generate a group of permutations of $\apols$.
If $\ell=1$, then it is not hard to show that $e^{\pi i/4} \in P_1$ has four distinct images under $\agr=\ggen{n,h,r}$, so $\agr$ has order at least $4$.
If $\ell >1$ and $\ell$ is odd, then it is not hard to show that $1+z+i z^{\ell-1} \in \apols$ has eight distinct images under $\agr=\ggen{n,h,r}$.
If $\ell$ is even, then it is not hard to show that $e^{\pi i/6} + e^{\pi i/3} z^{\ell-1} \in \apols$ has eight distinct images under $\agr=\ggen{n,h,r}$.
So $\agr$ has order at least $8$.
Furthermore $n$ commutes with both $h$ and $r$, and Lemma \ref{Alice}\eqref{Annie} shows that $h r=r h$ when $\ell$ is odd, but $h r = n r h$ when $\ell$ is even.

Thus if $\ell=1$, then $G_1=\ggen{n,h,r}=\ggen{n,r}$ is a group of order at least $4$, generated by commuting involutions $n$ and $r$.  So $G_1$ is the internal direct product of $\ggen{n}$ and $\ggen{r}$, which are both cyclic groups of order $2$.

If $\ell$ is odd and greater than $1$, our group $\agr=\ggen{n,h,r}$ is a group or order at least $8$ generated by commuting involutions $n$, $h$, and $r$.
So $\agr$ is the internal direct product of $\ggen{n}$, $\ggen{h}$, and $\ggen{r}$, which are three cyclic groups each of order $2$.

On the other hand, if $\ell$ is even, then $r h$ can be shown to have $(r h)^2=n$, $(r h)^3=n r h = h r$, and $(r h)^4$ the identity, and since the powers of $r h$ take the element $1+i z^{\ell-1} \in \apols$ to four distinct elements, we see that $r h$ has order $4$.
Thus $\agr=\ggen{n,h,r}=\ggen{h,r}=\ggen{h,r h}$.  Then note that $h(r h) h^{-1}= h r = (r h)^3=(r h)^{-1}$, and so it can be seen that $\agr=\ggen{h, r h}$ is generated by an element $h$ or order $2$ and an element $y= r h$ of order $4$ that satisfy the relation $h y h^{-1} = y^{-1}$.
These are the relations satisfied by the generators of the dihedral group $D$ of order $8$, that is, the group of symmetries of a square (with a $90^\circ$ rotation corresponding to $y$ and a flip corresponding to $h$).
So $\agr$ is a homomorphic image of $D$, but since $\agr$ has order at least $8$, we must have $\agr \cong D$.

To verify that $\norm{t(f)}{p}=\norm{f}{p}$ for any $f \in \apols$, $t \in \agr$, and $p \geq 1$, it suffices to check that it holds when $t$ is one of the generators $n$, $h$, and $r$.  When $t=n$, this is clear, and when $t=h$ or $r$, it is a consequence, respectively, of Lemma \ref{Leonard}\eqref{Norbert} or Lemma \ref{Orestes}.

Similarly, to verify that $\norm{t(f) \widetilde{t(f)}}{p}=\norm{f\tf}{p}$ for any $f \in \apols$, $t \in \agr$, and $p \geq 1$, it suffices to check that it holds when $t$ is one of the generators $n$, $h$, and $r$.  When $t=n$ or $h$, this is clear, and when $t=r$, then
\begin{align*}
\norm{r(f) \widetilde{r(f)}}{p}
& = \norm{f_0^\dag \widetilde{(f_0^\dag)}}{p} \\
& = \norm{f_0^\dag (\widetilde{f_0})^\dag}{p} \text{ or } \norm{-f_0^\dag (\widetilde{f_0})^\dag}{p} \\
& = \norm{(f_0 \widetilde{f_0})^\dag}{p} \\
& = \norm{f_0 \widetilde{f_0}}{p},
\end{align*}
where the second equality uses Lemma \ref{Alice}\eqref{Annie}, and the fourth equality uses Lemma \ref{Orestes}.
\end{proof}
\begin{corollary}\label{Wilbur}
Let $f_0 \in \C[z]$ be a polynomial of length $\ell$ with nonzero constant coefficient, let $t$ be an element of the group $\agr$ described in Proposition \ref{Penelope}, and let $a_0=t(f_0)$.
If $f_0,f_1,\ldots$ and $a_0,a_1,\ldots$ are sequences of Rudin-Shapiro-like polynomials generated from $f_0$ and $a_0$ via recursion \eqref{Francis}, then
\[
\lim_{n \to \infty} \ADF(a_n) = \lim_{n \to \infty} \ADF(f_n).
\]
\end{corollary}
\begin{proof}
By Theorem \ref{Vivian}, we have
\[
\lim_{n \to \infty} \ADF(f_n) = -1 + \frac{2}{3} \cdot \frac{\normff{f_0}+ \normtt{f_0 \tf_0}}{\normtf{f_0}},
\]
but Proposition \ref{Penelope} shows that the values of the three norms occurring on the right hand side do not change if we replace every instance of $f_0$ with $t(f_0)=a_0$, which changes the right hand side to $\lim_{n \to \infty} \ADF(a_n)$ by Theorem  \ref{Vivian}.
\end{proof}
Now we introduce a group of symmetries that, when applied to a pair of polynomials $(f_0(z),g_0(z))$ from $\C[z]$, will preserve the crosscorrelation properties of the stems $f_0,f_1,\ldots$ and $g_0,g_1,\ldots$ of Rudin-Shapiro-like polynomials obtained from seeds $f_0$ and $g_0$ via recursion \eqref{Francis}.  First we describe the group and how it affects certain norms and integrals, and then we show its effect on crosscorrelation as a corollary.
\begin{proposition}\label{Priscilla}
Let $\ell$ be a nonnegative integer, and let $\apols$ be the set of all polynomials of length $\ell$ in $\C[z]$ that have nonzero constant coefficient.
We define four maps from $\cpols$ to itself:
\begin{align*}
s(f,g) & = (g,f) \\
n(f,g) & = (-f,g) \\
h(f,g) & = (\tf,\tg) \\
r(f,g) & = (f^\dag,g^\dag).
\end{align*}
These maps generate a group $\cgr=\ggen{s,n,h,r}$ of permutations of $\cpols$.  $\cgr$ contains a dihedral subgroup $D$ of order $8$ generated $n s$ and $s$, where $n s$ has order $4$, $s$ has order $2$, and $s(n s)s^{-1}=(n s)^{-1}$.
\begin{enumerate}[(i).]
\item If $\ell=1$, then $\cgr$ is the internal direct product of the dihedral group $D$ of order $8$ and the cyclic group $\ggen{r}$ of order $2$.
\item If $\ell$ is odd and $\ell > 1$, $\cgr$ is the internal direct product of the dihedral group $D$ of order $8$, the cyclic subgroup $\ggen{h}$ of order $2$, and the cyclic subgroup $\ggen{r}$ of order $2$.
\item If $\ell$ is even, then $\cgr$ is the internal central product of $D$ and another dihedral subgroup $\Delta$ of order $8$ generated by $r h$ and $h$, where $r h$ has order $4$, $h$ has order $2$, and $h(r h) h^{-1}=(r h)^{-1}$.  Thus $\cgr$ is isomorphic to the extraspecial group of order $2^5$ of $+$ type, which is also the inner holomorph of the dihedral group of order $8$.  
\end{enumerate}
For any $t \in \cgr$, any $f,g \in \apols$, and any $p \geq 1$, let $(a,b)=t(f,g)$, and then we have
\begin{align*}
\norm{a}{p} \norm{b}{p} & =\norm{f}{p} \norm{g}{p}\\
\norm{a b}{p} & = \norm{f g}{p} \\
\norm{a\widetilde{b}}{p} & =\norm{f\tg}{p} \\
\Re \int a \widetilde{a} \conj{b \widetilde{b}} & = \Re \int f \tf \, \conj{g \tg}.
\end{align*}
\end{proposition}
\begin{proof}
It is clear that each of $s$, $n$, $h$, and $r$ is an involution on $\cpols$ (except that $h$ is the identity element when $\ell=1$), so these maps generate a group of permutations of $\cpols$.
Furthermore, $s$ commutes with $h$ and $r$, and $n$ also commutes with $h$ and $r$.
Thus we can better understand our group $\cgr$ by focusing on two subgroups, $\ggen{s,n}$ and $\ggen{h,r}$, with the knowledge that every element from the former subgroup commutes with every element of the latter subgroup.

Let us first focus on the subgroup $\ggen{s,n}$ of $\cgr$.
We note that $(s n)(f,g)=(g,-f)$ but $(n s)(f,g)=(-g,f)$, so that $s$ and $n$ do not commute.  We define $x=n s$, and then $\ggen{s,n}=\ggen{s, n s}=\ggen{s,x}$.  We note that $x$ is an element of order $4$ with $x^2(f,g)=(-f,-g)=-(f,g)$ and $x^3(f,g)=(g,-f)=(s n)(f,g)$.
Then note that $s x s^{-1}=s n s s^{-1}=s n=x^3=x^{-1}$.
Thus $\ggen{s,n}=\ggen{s,x}$ must be a homomorphic image of a dihedral group of order $8$, the group of symmetries of a square (with a $90^\circ$ rotation corresponding to $x$ and a flip corresponding to $s$).
And in fact, one can show that the $(1+z^{\ell-1},i+i z^{\ell-1}) \in \cpols$ has eight distinct images under the action of $\ggen{s,n}$, so $\ggen{s,x} \cong D$.

Now let us focus on the subgroup $\ggen{h,r}$ of $\cgr$.
If $\ell=1$, then $h$ is the identity element, so $\ggen{h,r}=\ggen{r}$ is a cyclic group of order $2$.
Since the elements of $\ggen{s,n}$ commute with the elements of $\ggen{h,r}$, this means that $G_{1,1}$ is a homomorphic image of a direct product of a dihedral group $D$ of order $8$ and a cyclic group $C$ of order $2$.
If $\ell=1$, then it is not hard to show that $(1,e^{\pi i/4}) \in P_1 \times P_1$ that has $16$ distinct images under $G_{1,1}=\ggen{n,h,r}$, so $G_{1,1} \cong D \times C$.

Now suppose that $\ell >1$.
Lemma \ref{Alice}\eqref{Annie} shows that $h r=r h$ when $\ell$ is odd, but that $(h r)(f,g) =-(r h)(f,g)$ when $\ell$ is even.
Note that the group element that maps $(f,g)$ to $(-f,-g)$ is $x^2$, described in the previous paragraph.
So $h$ and $r$ commute when $\ell$ is odd, but $h r = x^2 r h$ when $\ell$ is even.

So if $\ell>1$ and $\ell$ is odd, then $\ggen{h,r}$ is a homomorphic image of a Klein four-group, that is, of $C \times C$ with $C$ a cyclic group of order $2$.
And in fact one can show that $(1+z-z^{\ell-1},1+z-z^{\ell-1})$ has four distinct images under the action of $\ggen{h,r}$, so $\ggen{h,r}\cong C\times C$.

On the other hand, if $\ell$ is even, then $r h$ can be shown to have $(r h)^2=x^2$, $(r h)^3=x^2 r h = h r$, and $(r h)^4$ the identity, and since the powers of $r h$ take the element $(1+i z^{\ell-1},1+i z^{\ell-1}) \in \cpols$ to four distinct elements, we see that $r h$ has order $4$.
Note that $h(r h) h^{-1}= h r = (r h)^3=(r h)^{-1}$, and so it can be seen that $\ggen{h,r}=\ggen{h, r h}$ is generated by an element $h$ or order $2$ and an element $y=r h$ of order $4$ that satisfy the relation $h y h^{-1} = y^{-1}$.  Thus if $\ell$ is even, then $\ggen{r,h}$ is a homomorphic image of the dihedral group $D$ of order $8$, the group of symmetries of a square (with a $90^\circ$ rotation corresponding to $y$ and a flip corresponding to $h$).
One can show that $(1+i z^{\ell-1},1+i z^{\ell-1})$ has eight distinct images under the action of $\ggen{h,r}$, so that $\ggen{h,r}\cong D$.

Now we assemble what we have learned about the subgroups $\ggen{s,n}$ and $\ggen{h,r}$ of $\cgr$ using the fact that every element in the former subgroup commutes with every element in the latter.  We again separate into cases depending on the parity of $\ell$.

If $\ell$ is odd and greater than $1$, we saw that $\ggen{h,r}$ is a Klein four-group generated by the elements $h$ and $r$, each of order $2$.
We saw that $\ggen{s,n}$ is a dihedral group of order $8$.
So then the group $\cgr=\ggen{s,n,h,r}$ is a homomorphic image of $D \times C \times C$, where $D$ is the dihedral group of order $8$ and $C$ is the cyclic group of order $2$.
It is not hard to show that $(e^{\pi i/6} + e^{\pi i/3} z^{\ell-1},1+z+z^{\ell-1}) \in \cpols$ has $32$ distinct images under $\cgr=\ggen{n,h,r}$, so $\cgr$ has order at least $32$, and so we must have $\cgr \cong D \times C \times C$.

On the other hand, if $\ell$ is even, then $\ggen{h,r}$ is a dihedral group of order $8$ generated by element $y=r h$ of order $4$ and element $h$ of order $2$.
We saw that $\ggen{s,n}$ is a dihedral group of order $8$ generated by element $x=n s$ of order $4$ and element $s$ of order $2$.  The group $\ggen{s,n}=\ggen{x,s}$ has a center $\ggen{x^2}$ of order $2$.
The group $\ggen{h,r}=\ggen{y,h}$ has a center $\ggen{y^2}$, and we observed above that $y^2=(r h)^2=x^2$, so the centers of $\ggen{n,s}$ and $\ggen{h,r}$ completely overlap with each other.  So the group $\cgr=\ggen{s,n,h,r}$ is a homomorphic image of the central product of two dihedral groups of order $8$, which makes $\cgr$ a homomorphic image of the extraspecial group of order $2^5$ of $+$ type.
It is not hard to show that $(e^{\pi i/6} + e^{\pi i/3} z^{\ell-1},1+z^{\ell-1}) \in \cpols$ has $32$ distinct images under $\cgr=\ggen{n,h,r}$, so $\cgr$ has order at least $32$.
So $\cgr$ must be isomorphic to the the extraspecial group of order $2^5$ of $+$ type, which is also the inner holomorph of the dihedral group of order $8$.

Now we verify the four invariance relations for $\cgr$ in the statement of this proposition.  It suffices to check these relations when the group element $t \in \cgr$ is one of the four generators $s$, $n$, $h$, or $r$ of the group.  For the rest of this proof, we let $(f,g)$ be an arbitrary pair in $\cpols$, let $p$ be a real number with $p \geq 1$, let $t$ be an element of $\{s,n,h,r\}$, and we set $(a,b)=t(f,g)$.

It is clear that $\norm{a}{p} \norm{b}{p} =\norm{f}{p} \norm{g}{p}$ when $t=s$ or $n$, and when $t=h$ or $r$, this is a consequence, respectively, of Lemma \ref{Leonard}\eqref{Norbert} or Lemma \ref{Orestes}.

It is clear that $\norm{a b}{p} = \norm{f g}{p}$ when $t=s$ or $n$, and when $t=h$ or $r$, then this is a consequence, respectively, of Lemma \ref{Leonard}\eqref{Norbert} or Lemma \ref{Orestes}.

Now we verify that $\norm{a \tb}{p}=\norm{f\tg}{p}$.  This is clear when $t=n$, and when $t=s$ or $h$, then $\norm{a \tb}{p}=\norm{g \tf}{p}=\norm{\widetilde{f \tg}}{p}$, which equals $\norm{f\tg}{p}$ by Lemma \ref{Leonard}\eqref{Norbert}.  When $t=r$, then
\begin{align*}
\norm{a \widetilde{b}}{p}
& = \norm{f^\dag \widetilde{(g^\dag)}}{p} \\
& = \norm{f^\dag (\tg)^\dag}{p} \text{ or } \norm{-f^\dag (\tg)^\dag}{p} \\
& = \norm{(f \tg)^\dag}{p} \\
& = \norm{f \tg}{p},
\end{align*}
where Lemmata \ref{Alice}\eqref{Annie} and \ref{Orestes} are used in the second and fourth equalities.

Now we verify that $\Re \int a \widetilde{a} \conj{b \widetilde{b}} = \Re \int f \tf \, \conj{g \tg}$.
When $t=n$ or $h$, this is clear, and when $t=s$, we see that $\Re \int a \widetilde{a} \conj{b \widetilde{b}}=\Re \int g \tg \conj{f \tf}$, which is equal to $\Re \int f \tf \, \conj{g \tg}$, since conjugation of the integrand does not change the real part of the integral.
Finally, if $t=r$, then
\begin{align*}
\Re \int a \ta \conj{b \tb}
& = \Re \int f^\dag \widetilde{f^\dag} \conj{g^\dag \widetilde{g^\dag}} \\
& = \Re \int f^\dag (\tf)^\dag \conj{g^\dag (\tg)^\dag}\\
& = \Re \int \conj{\conj{(f \tf)^\dag}} \conj{(g \tg)^\dag} \\
& = \Re \int \left(\conj{z^{-2\deg f} f \tf}\right) \left(z^{-2\deg g} g \tg\right) \\
& = \Re \int \conj{f \tf} g \tg \\
& = \Re \int f \tf \, \conj{g \tg},
\end{align*}
where the second equality uses Lemma \ref{Alice}\eqref{Annie} (and the fact that $f$ and $g$ are assumed to have the same degree), the fourth equality uses Lemma \ref{Alice}\eqref{Alexander}, the fifth equality uses the fact that $\conj{z}=z^{-1}$ on the complex unit circle, and the last equality uses the fact that conjugation of the integrand does not change the real part of the integral.
\end{proof}
\begin{corollary}\label{Eric}
Let $f_0, g_0 \in \C[z]$ be a polynomials of length $\ell$ with nonzero constant coefficients, let $t$ be an element of the group $\cgr$ be the group described in Proposition \ref{Priscilla}, and let $(a_0,b_0)=t(f_0,g_0)$.
If $f_0,f_1,\ldots$ and $g_0,g_1,\ldots$ and $a_0,a_1,\ldots$ and $b_0,b_1,\ldots$ are sequences of Rudin-Shapiro-like polynomials generated from $f_0$, $g_0$, $a_0$, and $b_0$ via recursion \eqref{Francis}, then
\begin{align*}
\lim_{n \to \infty} \CDF(a_n,b_n) & = \lim_{n \to \infty} \CDF(f_n,g_n) \text{ and}\\
\lim_{n \to \infty} \PSC(a_n,b_n) & = \lim_{n \to \infty} \PSC(f_n,g_n).
\end{align*}
\end{corollary}
\begin{proof}
By Theorem \ref{Vivian}, we have
\[
\lim_{n \to \infty} \CDF(f_n,g_n) = \frac{2 \normtt{f_0 g_0}+\normtt{f_0 \tg_0} + \Re  \int f_0 \tf_0 \conj{g_0 \tg_0}}{3 \normtt{f_0} \normtt{g_0}},
\]
but Proposition \ref{Priscilla} shows that the values of the three terms in the numerator and the value of the denominator of the right hand side do not change if we replace every instance of $f_0$ with $a_0$ and every instance of $g_0$ with $b_0$.  These replacements change the right hand side to $\lim_{n \to \infty} \CDF(a_n,b_n)$ by Theorem \ref{Vivian}.

Because of the structure of the groups $\agr$ and $\cgr$ described in Propositions \ref{Penelope} and \ref{Priscilla} above, one can say that there exist $u,v \in \agr$ such that either $(a_0,b_0)=(u(f_0),v(g_0))$ or $(a_0,b_0)=(v(g_0),u(f_0))$.  Thus by Corollary \ref{Wilbur},
\[
\lim_{n \to \infty} \ADF(a_n) \ADF(b_n) = \lim_{n \to \infty} \ADF(f_n) \ADF(g_n),
\]
and so, considering the formula \eqref{Paul} for the Pursley-Sarwate Criterion, we see that
\[
\lim_{n \to \infty} \PSC(a_n,b_n) = \lim_{n \to \infty} \PSC(f_n,g_n).\qedhere
\]
\end{proof}

\section{Some Examples of Pairs of Rudin-Shapiro-Like Sequences with Low Correlation}\label{Thomas}

For each $\ell \leq 52$, we considered every possible Littlewood polynomial $f_0$ of length $\ell$, and used computers, including opportunistic use of distributed resources through the Open Science Grid \cite{OSG-1,OSG-2}, to calculate via Corollary \ref{Barbara} (a result originally due to Borwein and Mossinghoff \cite[Theorem 1]{Borwein-Mossinghoff}) the limiting autocorrelation demerit factor of the stem $f_0,f_1,\ldots$ constructed from seed $f_0$ via our recursion \eqref{Francis}.  For each length $\ell$, we report on Table \ref{Timothy} the lowest limiting autocorrelation demerit factor achieved, and indicate how many seeds achieve this minimum value.  Seeds that are equivalent modulo the action of the group $\agr$ described in Proposition \ref{Penelope} always have the same limiting autocorrelation demerit factor by Corollary \ref{Wilbur}, and so we group seeds into orbits under the action of $\agr$ and report how many distinct orbits there are on Table \ref{Timothy}.
\renewcommand{\arraystretch}{1.3}%
\begin{table}[!ht]
\caption{Lowest Limiting Autocorrelation Demerit Factor for Seeds of Each Length}\label{Timothy}
\begin{center}
\begin{tabular}{|c|r|cc|c|}
\hline
seed & \multicolumn{1}{c|}{limiting} & number of & number of & sample \\
length & \multicolumn{1}{c|}{$\ADF(f_n)$} & sequences & orbits & seed $f_0$ \\ \hline\hline
$1$ & $\frac{1}{3}=0.3333\ldots$ & $2$ & $1$ & {\tt 2} \\ \hline
$2$ & $\frac{1}{3}=0.3333\ldots$ & $4$ & $1$ & {\tt 0} \\ \hline
$3$ & $\frac{17}{27}=0.6296\ldots$ & $8$ & $2$ & {\tt 0} \\ \hline
$4$ & $\frac{1}{3}=0.3333\ldots$ & $8$ & $1$ & {\tt 1} \\ \hline
$5$ & $\frac{41}{75}=0.5466\ldots$ & $24$ & $4$ & {\tt 01} \\ \hline
$6$ & $\frac{17}{27}=0.6296\ldots$ & $56$ & $8$ & {\tt 01} \\ \hline
$7$ & $\frac{73}{147}=0.4965\ldots$ & $56$ & $8$ & {\tt 03} \\ \hline
$8$ & $\frac{1}{3}=0.3333\ldots$ & $32$ & $4$ & {\tt 06} \\ \hline
$9$ & $\frac{113}{243}=0.4650\ldots$ & $144$ & $18$ & {\tt 006} \\ \hline
$10$ & $\frac{41}{75}=0.5466\ldots$ & $504$ & $64$ & {\tt 006} \\ \hline
$11$ & $\frac{161}{363}=0.4435\ldots$ & $168$ & $22$ & {\tt 01C} \\ \hline
$12$ & $\frac{11}{27}=0.4074\ldots$ & $96$ & $12$ & {\tt 036} \\ \hline
$13$ & $\frac{217}{507}=0.4280\ldots$ & $344$ & $44$ & {\tt 0036} \\ \hline
$14$ & $\frac{73}{147}=0.4965\ldots$ & $2648$ & $332$ & {\tt 0036} \\ \hline
$15$ & $\frac{281}{675}=0.4162\ldots$ & $688$ & $86$ & {\tt 0163} \\ \hline
$16$ & $\frac{1}{3}=0.3333\ldots$ & $192$ & $24$ & {\tt 0359} \\ \hline
$17$ & $\frac{353}{867}=0.4071\ldots$ & $1472$ & $184$ & {\tt 001C9} \\ \hline
$18$ & $\frac{113}{243}=0.4650\ldots$ & $12992$ & $1624$ & {\tt 001C9} \\ \hline
$19$ & $\frac{433}{1083}=0.3998\ldots$ & $784$ & $98$ & {\tt 00793} \\ \hline
$20$ & $\frac{1}{3}=0.3333\ldots$ & $128$ & $16$ & {\tt 05239} \\ \hline
$21$ & $\frac{521}{1323}=0.3938\ldots$ & $1312$ & $164$ & {\tt 000F19} \\ \hline
$22$ & $\frac{161}{363}=0.4435\ldots$ & $35352$ & $4420$ & {\tt 000F19} \\ \hline
$23$ & $\frac{617}{1587}=0.3887\ldots$ & $1696$ & $212$ & {\tt 0066B4} \\ \hline
$24$ & $\frac{19}{54}=0.3518\ldots$ & $320$ & $40$ & {\tt 00CD69} \\ \hline
$25$ & $\frac{721}{1875}=0.3845\ldots$ & $2176$ & $272$ & {\tt 000CD29} \\ \hline
$26$ & $\frac{217}{507}=0.4280\ldots$ & $104920$ & $13116$ & {\tt 0007866} \\ \hline
\end{tabular}
\end{center}
\end{table}%
\setcounter{table}{0}%
\begin{table}[!ht]
\caption{(continued) Lowest Limiting Autocorrelation Demerit Factor for Seeds of Each Length}
\begin{center}
\begin{tabular}{|c|r|cc|c|}
\hline
seed & \multicolumn{1}{c|}{limiting} & number of & number of & sample \\
length & \multicolumn{1}{c|}{$\ADF(f_n)$} & sequences & orbits & seed $f_0$ \\ \hline\hline
$27$ & $\frac{833}{2187}=0.3808\ldots$ & $1888$ & $236$ & {\tt 006B274} \\ \hline
$28$ & $\frac{53}{147}=0.3605\ldots$ & $512$ & $64$ & {\tt 00DB171} \\ \hline
$29$ & $\frac{953}{2523}=0.3777\ldots$ & $3040$ & $380$ & {\tt 000E4B8D} \\ \hline
$30$ & $\frac{281}{675}=0.4162\ldots$ & $266688$ & $33336$ & {\tt 0006C729} \\ \hline
$31$ & $\frac{1081}{2883}=0.3749\ldots$ & $6368$ & $796$ & {\tt 001E2D33} \\ \hline
$32$ & $\frac{1}{3}=0.3333\ldots$ & $1536$ & $192$ & {\tt 003C5A66} \\ \hline
$33$ & $\frac{1217}{3267}=0.3725\ldots$ & $10400$ & $1300$ & {\tt 0003C5A66} \\ \hline
$34$ & $\frac{353}{867}=0.4071\ldots$ & $554752$ & $69344$ & {\tt 0003C5A66} \\ \hline
$35$ & $\frac{1361}{3675}=0.3703\ldots$ & $1216$ & $152$ & {\tt 001F1699C} \\ \hline
$36$ & $\frac{29}{81}=0.3580\ldots$ & $640$ & $80$ & {\tt 0034EC5A6} \\ \hline
$37$ & $\frac{1513}{4107}=0.3683\ldots$ & $1760$ & $220$ & {\tt 00035AC726} \\ \hline
$38$ & $\frac{433}{1083}=0.3998\ldots$ & $840256$ & $105032$ & {\tt 00034E94E6} \\ \hline
$39$ & $\frac{1673}{4563}=0.3666\ldots$ & $4416$ & $552$ & {\tt 0019E2D2B3} \\ \hline
$40$ & $\frac{1}{3}=0.3333\ldots$ & $1088$ & $136$ & {\tt 0033C5A566} \\ \hline
$41$ & $\frac{1841}{5043}=0.3650\ldots$ & $7328$ & $916$ & {\tt 00033C5A566} \\ \hline
$42$ & $\frac{521}{1323}=0.3938\ldots$ & $1589568$ & $198696$ & {\tt 0001E5A3599} \\ \hline
$43$ & $\frac{2017}{5547}=0.3636\ldots$ & $2592$ & $324$ & {\tt 0015B878CCB} \\ \hline
$44$ & $\frac{125}{363}=0.3443\ldots$ & $256$ & $32$ & {\tt 00178B4B326} \\ \hline
$45$ & $\frac{2201}{6075}=0.3623\ldots$ & $2272$ & $284$ & {\tt 0001E9663D33} \\ \hline
$46$ & $\frac{617}{1587}=0.3887\ldots$ & $2690528$ & $336316$ & {\tt 0000FC31E199} \\ \hline
$47$ & $\frac{2393}{6627}=0.3610\ldots$ & $2752$ & $344$ & {\tt 0006E529E49C} \\ \hline
$48$ & $\frac{73}{216}=0.3379\ldots$ & $128$ & $16$ & {\tt 003C3315A9A6} \\ \hline
$49$ & $\frac{2593}{7203}=0.3599\ldots$ & $2720$ & $340$ & {\tt 0000F30F4A665} \\ \hline
$50$ & $\frac{721}{1875}=0.3845\ldots$ & $3751392$ & $468924$ & {\tt 0000DD83C6696} \\ \hline
$51$ & $\frac{2801}{7803}=0.3589\ldots$ & $1536$ & $192$ & {\tt 0006F1C6D2372} \\ \hline
$52$ & $\frac{1}{3}=0.3333\ldots$ & $64$ & $8$ & {\tt 00C3CC459A96A} \\ \hline
\end{tabular}
\end{center}
\end{table}%
\renewcommand{\arraystretch}{1.0}%
This distributed computational search summarized Table \ref{Timothy} used roughly 400,000 hours of wall-clock time for the processors of our own and those of the Open Science Grid.

For each length $\ell$, we give one example of a seed $f_0$ whose stem achieves the smallest limiting autocorrelation demerit factor.  Example seeds are reported using a hexadecimal code.  To decode, expand each hexadecimal digit into binary form ({\tt 0} $\to$ {\tt 0000}, {\tt 1} $\to$ {\tt 0001}, $\ldots$, {\tt F} $\to$ {\tt 1111}) and, if necessary, remove initial {\tt 0} symbols to obtain a binary sequence of the appropriate length.  Then convert each {\tt 0} to $+1$ and each {\tt 1} to $-1$ to obtain the list of coefficients of the seed $f_0$.  For example, Table \ref{Mary} reports for length $\ell=14$ that one seed of interest is {\tt 149B}.  Expand to {\tt 0001\,\,0100\,\,1001\,\,1011} and delete the initial two zeroes to obtain a sequence {\tt 01\,\,0100\,\,1001\,\,1011} of length $\ell=14$.  Convert from $0,1$ to $\pm 1$ to obtain the coefficients of 
\[
g_0(z)=1-z+z^2-z^3+z^4+z^5-z^6+z^7+z^8-z^9-z^{10}+z^{11}-z^{12}-z^{13}.
\]

Borwein-Mossinghoff \cite[Corollary 1]{Borwein-Mossinghoff} proved that the limiting autocorrelation demerit factor for any sequence $f_0,f_1,\ldots$ of Rudin-Shapiro-like Littlewood polynomials generated from a seed $f_0$ using recursion \eqref{Francis} can never be less than $1/3$ (see also our Theorem \ref{Vivian}).   Their computer experiments show that for seeds of length $\ell \leq 40$, there exist seeds whose stems achieve limiting autocorrelation demerit factor $1/3$ if and only if $\ell \in \{1,2,4,8,16,20,32,40\}$.  We also discovered seeds of length $52$ whose stems achieve limiting autocorrelation demerit factor $1/3$.  The first and third authors have now proved \cite{Katz-Trunov} that a seed $f_0$ of length $\ell > 1$ produces a stem with limiting autocorrelation demerit factor $1/3$ if and only if $f_0$ is the interleaving of the two sequences of some Golay complementary pair.  This explains why both Borwein and Mossinghoff's searches and ours produced seeds with optimal asymptotic autocorrelation at the lengths that we have observed.  For lengths $\ell \leq 40$, where there are seeds with optimal asymptotic correlation, Borwein and Mossinghoff also indicate how many distinct optimal seeds there are for each such length.  Our computer experiments agree with theirs, but we also present on Table \ref{Timothy} the minimum limiting autocorrelation demerit factors for all lengths $\ell \leq 52$, regardless of whether or not the minimum is $1/3$.

Now let us also consider crosscorrelation.
In view of Proposition \ref{Elaine} and Remark \ref{Raphael}, we already know we can achieve limiting crosscorrelation demerit factors as close to $0$ as we like, but we have observed that pairs of stems with very low limiting crosscorrelation demerit factor tend to have very high autocorrelation demerit factor, and are therefore of little practical value.  This is not surprising, given the bound \eqref{Hyeon} of Pursley and Sarwate.
It is much more enlightening to ask how low one can make the limiting Pursley-Sarwate Criterion \eqref{Paul}, which combines both autocorrelation and crosscorrelation performance.

For each $\ell \leq 28$, we considered every possible pair of Littlewood polynomials $(f_0,g_0)$ of length $\ell$, and used computers, including opportunistic use of distributed resources through the Open Science Grid \cite{OSG-1,OSG-2}, to calculate via Theorem \ref{Vivian} the limiting crosscorrelation demerit factors of the pair of stems $(f_0,f_1,\ldots;g_0,g_1,\ldots)$ constructed from our recursion \eqref{Francis}.  We also calculate the limiting autocorrelation demerit factors for each of the two stems, and from all three of these limits, we obtain the limiting Pursley-Sarwate Criterion.  Table \ref{Mary} records the lowest limiting Pursley-Sarwate Criterion achieved for each $\ell \leq 28$, and records the seed pairs $(f_0,g_0)$ that give rise to the pairs of stems that achieve this minimum.
Seed pairs that are equivalent modulo the action of the group $\cgr$ described in Proposition \ref{Priscilla} will always have the same limiting Pursley-Sarwate Criterion by Corollary \ref{Eric}, and so we group seeds pairs into orbits under the action of $\cgr$.  We report one representative of each class on Table \ref{Mary} using our hexadecimal code (described above in the discussion of Table \ref{Timothy}) and also report the size of the orbit.  For some lengths there are multiple equivalence classes that achieve the same minimum limiting Pursley-Sarwate Criterion: each such class has its own line on the table.
\renewcommand{\arraystretch}{1.3}%
\begin{table}[!ht]
\caption{Lowest Limiting Pursley-Sarwate Criterion for Seeds of Each Length}\label{Mary}
\begin{center}
\begin{tabular}{|c|lccc|c|c|c|}
\hline
seed   & \multicolumn{4}{|c|}{limiting values as $n\to\infty$} & orbit & \multicolumn{2}{c|}{seeds} \\
length & \multicolumn{1}{c}{\tiny{$\PSC(f_n,g_n)$}} & {\tiny $\ADF(f_n)$} & {\tiny $\ADF(g_n)$} & {\tiny $\CDF(f_n,g_n)$} & size & $f_0$ & $g_0$  \\
\hline\hline
$1$ & $1.6666\ldots$ & $\frac{1}{3}$ & $\frac{1}{3}$ & $\frac{4}{3}$ & $4$ & \stt{0} & \stt{0} \\\hline
$2$ & $1.3333\ldots$ & $\frac{1}{3}$ & $\frac{1}{3}$ & $\frac{1}{1}$ & $8$ & \stt{0} & \stt{1} \\\hline
$3$ & $1.3703\ldots$ & $\frac{17}{27}$ & $\frac{17}{27}$ & $\frac{20}{27}$ & $32$ & \stt{0} & \stt{1} \\\hline
$4$ & $1.1666\ldots$ & $\frac{1}{3}$ & $\frac{1}{3}$ & $\frac{5}{6}$ & $16$ & \stt{1} & \stt{2} \\\hline
$5$ & $1.3466\ldots$ & $\frac{41}{75}$ & $\frac{41}{75}$ & $\frac{4}{5}$ & $32$ & \stt{01} & \stt{02} \\
$5$ & $1.3466\ldots$ & $\frac{41}{75}$ & $\frac{41}{75}$ & $\frac{4}{5}$ & $32$ & \stt{01} & \stt{08} \\
$5$ & $1.3466\ldots$ & $\frac{41}{75}$ & $\frac{41}{75}$ & $\frac{4}{5}$ & $32$ & \stt{01} & \stt{0D} \\\hline
$6$ & $1.2962\ldots$ & $\frac{17}{27}$ & $\frac{17}{27}$ & $\frac{2}{3}$ & $32$ & \stt{02} & \stt{0D} \\
$6$ & $1.2962\ldots$ & $\frac{17}{27}$ & $\frac{17}{27}$ & $\frac{2}{3}$ & $32$ & \stt{04} & \stt{0B} \\\hline
$7$ & $1.2312\ldots$ & $\frac{73}{147}$ & $\frac{73}{147}$ & $\frac{36}{49}$ & $32$ & \stt{04} & \stt{1A} \\\hline
$8$ & $1.1666\ldots$ & $\frac{1}{3}$ & $\frac{1}{3}$ & $\frac{5}{6}$ & $32$ & \stt{06} & \stt{3A} \\
$8$ & $1.1666\ldots$ & $\frac{1}{3}$ & $\frac{1}{3}$ & $\frac{5}{6}$ & $32$ & \stt{12} & \stt{2E} \\\hline
$9$ & $1.2057\ldots$ & $\frac{113}{243}$ & $\frac{113}{243}$ & $\frac{20}{27}$ & $32$ & \stt{009} & \stt{035} \\\hline
$10$ & $1.1733\ldots$ & $\frac{41}{75}$ & $\frac{41}{75}$ & $\frac{47}{75}$ & $32$ & \stt{04D} & \stt{0A1} \\\hline
$11$ & $1.1818\ldots$ & $\frac{161}{363}$ & $\frac{161}{363}$ & $\frac{268}{363}$ & $32$ & \stt{032} & \stt{251} \\\hline
$12$ & $1.1666\ldots$ & $\frac{11}{27}$ & $\frac{11}{27}$ & $\frac{41}{54}$ & $32$ & \stt{065} & \stt{6A3} \\\hline
$13$ & $1.1734\ldots$ & $\frac{217}{507}$ & $\frac{281}{507}$ & $\frac{116}{169}$ & $32$ & \stt{00CA} & \stt{03AD} \\
$13$ & $1.1734\ldots$ & $\frac{217}{507}$ & $\frac{281}{507}$ & $\frac{116}{169}$ & $32$ & \stt{00CA} & \stt{0907} \\\hline
$14$ & $1.1836\ldots$ & $\frac{73}{147}$ & $\frac{73}{147}$ & $\frac{101}{147}$ & $32$ & \stt{0071} & \stt{149B} \\\hline
\end{tabular}
\end{center}
\end{table}%
\setcounter{table}{1}%
\begin{table}[!ht]
\caption{(continued) Lowest Limiting Pursley-Sarwate Criterion for Seeds of Each Length}
\begin{center}
\begin{tabular}{|c|lccc|c|c|c|}
\hline
seed   & \multicolumn{4}{|c|}{limiting values as $n\to\infty$} & orbit & \multicolumn{2}{c|}{seeds} \\
length & \multicolumn{1}{c}{\tiny{$\PSC(f_n,g_n)$}} & {\tiny $\ADF(f_n)$} & {\tiny $\ADF(g_n)$} & {\tiny $\CDF(f_n,g_n)$} & size & $f_0$ & $g_0$  \\
\hline\hline
$15$ & $1.1546\ldots$ & $\frac{281}{675}$ & $\frac{23}{45}$ & $\frac{52}{75}$ & $32$ & \stt{024E} & \stt{15C3} \\\hline
$16$ & $1.1041\ldots$ & $\frac{1}{3}$ & $\frac{1}{3}$ & $\frac{37}{48}$ & $32$ & \stt{0A36} & \stt{11D2} \\\hline
$17$ & $1.1407\ldots$ & $\frac{353}{867}$ & $\frac{353}{867}$ & $\frac{212}{289}$ & $32$ & \stt{0038D} & \stt{0EE96} \\\hline
$18$ & $1.1481\ldots$ & $\frac{113}{243}$ & $\frac{113}{243}$ & $\frac{166}{243}$ & $32$ & \stt{0039A} & \stt{0E8F6} \\\hline
$19$ & $1.1559\ldots$ & $\frac{433}{1083}$ & $\frac{497}{1083}$ & $\frac{788}{1083}$ & $32$ & \stt{00E4D} & \stt{38A16} \\
$19$ & $1.1559\ldots$ & $\frac{433}{1083}$ & $\frac{497}{1083}$ & $\frac{788}{1083}$ & $32$ & \stt{0C56D} & \stt{0E013} \\\hline
$20$ & $1.1363\ldots$ & $\frac{11}{25}$ & $\frac{1}{3}$ & $\frac{113}{150}$ & $32$ & \stt{08FA6} & \stt{5A230} \\\hline
$21$ & $1.1338\ldots$ & $\frac{521}{1323}$ & $\frac{65}{147}$ & $\frac{316}{441}$ & $32$ & \stt{00F765} & \stt{05DAF3} \\\hline
$22$ & $1.1515\ldots$ & $\frac{161}{363}$ & $\frac{161}{363}$ & $\frac{257}{363}$ & $32$ & \stt{0188B5} & \stt{1341DE} \\
$22$ & $1.1515\ldots$ & $\frac{161}{363}$ & $\frac{161}{363}$ & $\frac{257}{363}$ & $32$ & \stt{022D85} & \stt{0C74FD} \\\hline
$23$ & $1.1203\ldots$ & $\frac{745}{1587}$ & $\frac{617}{1587}$ & $\frac{1100}{1587}$ & $32$ & \stt{0BA421} & \stt{376A38} \\\hline
$24$ & $1.1292\ldots$ & $\frac{23}{54}$ & $\frac{7}{18}$ & $\frac{13}{18}$ & $32$ & \stt{02B25C} & \stt{7A8C2C} \\\hline
$25$ & $1.1372\ldots$ & $\frac{157}{375}$ & $\frac{721}{1875}$ & $\frac{92}{125}$ & $32$ & \stt{00A9273} & \stt{0BFC9C7} \\\hline
$26$ & $1.1350\ldots$ & $\frac{217}{507}$ & $\frac{83}{169}$ & $\frac{343}{507}$ & $32$ & \stt{075D9AD} & \stt{1ACFF83} \\\hline
$27$ & $1.1323\ldots$ & $\frac{961}{2187}$ & $\frac{299}{729}$ & $\frac{172}{243}$ & $32$ & \stt{014E48A} & \stt{03A3DE6} \\\hline
$28$ & $1.1258\ldots$ & $\frac{59}{147}$ & $\frac{59}{147}$ & $\frac{71}{98}$ & $32$ & \stt{09467C5} & \stt{60EA253} \\\hline
\end{tabular}
\end{center}
\end{table}%
\renewcommand{\arraystretch}{1.0}%
This distributed computational search summarized Table \ref{Mary} used roughly 100,000 hours of wall-clock time for the processors of our own and those of the Open Science Grid.

In Table \ref{Andrew} we also present sequences with low limiting Pursley-Sarwate Criterion.
For a given length $\ell \leq 52$, we first found all the seeds that produce stems whose limiting autocorrelation demerit factors reach the minimum value for that length, as reported in Table \ref{Timothy}.  Then we compute which pairs of these seeds $(f_0,g_0)$ produce pairs of stems with the lowest limiting crosscorrelation demerit factor (and therefore the lowest limiting Pursley-Sarwate Criterion, since they all have the same limiting autocorrelation demerit factors).
Given the data file that records the seeds of a given length whose stems have minimum limiting autocorrelation demerit factor (produced while compiling Table \ref{Timothy}), the task of computing the lowest limiting Pursley-Sarwate Criterion among stems produced from these seeds took very little time (around a minute at most, and usually much less).
Seed pairs that are equivalent modulo the action of the group $\cgr$ described in Proposition \ref{Priscilla} will always have the same limiting Pursley-Sarwate Criterion by Corollary \ref{Eric}, and so we group seeds pairs into orbits under the action of $\cgr$.  We report one representative of each class on Table \ref{Andrew} using our hexadecimal code (described above in the discussion of Table \ref{Timothy}) and also report the size of the orbit.  For some lengths there are multiple equivalence classes that achieve the same minimum limiting Pursley-Sarwate Criterion: each such class has its own line on the table.
\renewcommand{\arraystretch}{1.3}%
\begin{table}[!ht]
\caption{Lowest Limiting Pursley-Sarwate Criterion among Seed Pairs that Have the Lowest Limiting Autocorrelation Demerit Factor}\label{Andrew}
\begin{center}
\begin{tabular}{|c|lccc|c|c|c|}
\hline
\footnotesize{seed}  & \multicolumn{4}{|c|}{limiting values as $n\to\infty$} & \footnotesize{orbit} & \multicolumn{2}{c|}{{\footnotesize seeds}} \\
\footnotesize{length} & \multicolumn{1}{c}{\tiny{$\PSC(f_n,g_n)$}} & {\tiny $\ADF(f_n)$} & {\tiny $\ADF(g_n)$} & {\tiny $\CDF(f_n,g_n)$} & \footnotesize{size} & {\footnotesize $f_0$} & {\footnotesize $g_0$}  \\
\hline\hline
$1$ & $1.6666\ldots$ & $\frac{1}{3}$ & $\frac{1}{3}$ & $\frac{4}{3}$ & $4$ & \ttt{0} & \ttt{0} \\\hline
$2$ & $1.3333\ldots$ & $\frac{1}{3}$ & $\frac{1}{3}$ & $\frac{1}{1}$ & $8$ & \ttt{0} & \ttt{1} \\\hline
$3$ & $1.3703\ldots$ & $\frac{17}{27}$ & $\frac{17}{27}$ & $\frac{20}{27}$ & $32$ & \ttt{0} & \ttt{1} \\\hline
$4$ & $1.1666\ldots$ & $\frac{1}{3}$ & $\frac{1}{3}$ & $\frac{5}{6}$ & $16$ & \ttt{1} & \ttt{2} \\\hline
$5$ & $1.3466\ldots$ & $\frac{41}{75}$ & $\frac{41}{75}$ & $\frac{4}{5}$ & $32$ & \ttt{01} & \ttt{02} \\
$5$ & $1.3466\ldots$ & $\frac{41}{75}$ & $\frac{41}{75}$ & $\frac{4}{5}$ & $32$ & \ttt{01} & \ttt{08} \\
$5$ & $1.3466\ldots$ & $\frac{41}{75}$ & $\frac{41}{75}$ & $\frac{4}{5}$ & $32$ & \ttt{01} & \ttt{0D} \\\hline
$6$ & $1.2962\ldots$ & $\frac{17}{27}$ & $\frac{17}{27}$ & $\frac{2}{3}$ & $32$ & \ttt{02} & \ttt{0D} \\
$6$ & $1.2962\ldots$ & $\frac{17}{27}$ & $\frac{17}{27}$ & $\frac{2}{3}$ & $32$ & \ttt{04} & \ttt{0B} \\\hline
$7$ & $1.2312\ldots$ & $\frac{73}{147}$ & $\frac{73}{147}$ & $\frac{36}{49}$ & $32$ & \ttt{04} & \ttt{1A} \\\hline
$8$ & $1.1666\ldots$ & $\frac{1}{3}$ & $\frac{1}{3}$ & $\frac{5}{6}$ & $32$ & \ttt{06} & \ttt{3A} \\
$8$ & $1.1666\ldots$ & $\frac{1}{3}$ & $\frac{1}{3}$ & $\frac{5}{6}$ & $32$ & \ttt{12} & \ttt{2E} \\\hline
$9$ & $1.2057\ldots$ & $\frac{113}{243}$ & $\frac{113}{243}$ & $\frac{20}{27}$ & $32$ & \ttt{009} & \ttt{035} \\\hline
$10$ & $1.1733\ldots$ & $\frac{41}{75}$ & $\frac{41}{75}$ & $\frac{47}{75}$ & $32$ & \ttt{04D} & \ttt{0A1} \\\hline
$11$ & $1.1818\ldots$ & $\frac{161}{363}$ & $\frac{161}{363}$ & $\frac{268}{363}$ & $32$ & \ttt{032} & \ttt{251} \\\hline
$12$ & $1.1666\ldots$ & $\frac{11}{27}$ & $\frac{11}{27}$ & $\frac{41}{54}$ & $32$ & \ttt{065} & \ttt{6A3} \\\hline
$13$ & $1.1775\ldots$ & $\frac{217}{507}$ & $\frac{217}{507}$ & $\frac{380}{507}$ & $32$ & \ttt{01DB} & \ttt{0D47} \\\hline
$14$ & $1.1836\ldots$ & $\frac{73}{147}$ & $\frac{73}{147}$ & $\frac{101}{147}$ & $32$ & \ttt{0071} & \ttt{149B} \\\hline
$15$ & $1.1688\ldots$ & $\frac{281}{675}$ & $\frac{281}{675}$ & $\frac{508}{675}$ & $32$ & \ttt{01AC} & \ttt{245C} \\\hline
$16$ & $1.1041\ldots$ & $\frac{1}{3}$ & $\frac{1}{3}$ & $\frac{37}{48}$ & $32$ & \ttt{0A36} & \ttt{11D2} \\\hline
$17$ & $1.1407\ldots$ & $\frac{353}{867}$ & $\frac{353}{867}$ & $\frac{212}{289}$ & $32$ & \ttt{0038D} & \ttt{0EE96} \\\hline
$18$ & $1.1481\ldots$ & $\frac{113}{243}$ & $\frac{113}{243}$ & $\frac{166}{243}$ & $32$ & \ttt{0039A} & \ttt{0E8F6} \\\hline
$19$ & $1.1643\ldots$ & $\frac{433}{1083}$ & $\frac{433}{1083}$ & $\frac{276}{361}$ & $32$ & \ttt{00F26} & \ttt{0C549} \\\hline
$20$ & $1.14$ & $\frac{1}{3}$ & $\frac{1}{3}$ & $\frac{121}{150}$ & $16$ & \ttt{05239} & \ttt{36E0A} \\\hline
$21$ & $1.1405\ldots$ & $\frac{521}{1323}$ & $\frac{521}{1323}$ & $\frac{988}{1323}$ & $32$ & \ttt{001C9A} & \ttt{063EAD} \\\hline
$22$ & $1.1515\ldots$ & $\frac{161}{363}$ & $\frac{161}{363}$ & $\frac{257}{363}$ & $32$ & \ttt{0188B5} & \ttt{1341DE} \\
$22$ & $1.1515\ldots$ & $\frac{161}{363}$ & $\frac{161}{363}$ & $\frac{257}{363}$ & $32$ & \ttt{022D85} & \ttt{0C74FD} \\\hline
\end{tabular}
\end{center}
\end{table}%
\setcounter{table}{2}%
\begin{table}[!ht]
\caption{(continued) Lowest Limiting Pursley-Sarwate Criterion among Seed Pairs that Have the Lowest Limiting Autocorrelation Demerit Factor}
\vspace{-0.9mm}
\begin{center}
\begin{tabular}{|c|lccc|c|c|c|}
\hline
\footnotesize{seed}  & \multicolumn{4}{|c|}{limiting values as $n\to\infty$} & \footnotesize{orbit} & \multicolumn{2}{c|}{{\footnotesize seeds}} \\
\footnotesize{length} & \multicolumn{1}{c}{\tiny{$\PSC(f_n,g_n)$}} & {\tiny $\ADF(f_n)$} & {\tiny $\ADF(g_n)$} & {\tiny $\CDF(f_n,g_n)$} & \footnotesize{size} & {\footnotesize $f_0$} & {\footnotesize $g_0$}  \\
\hline\hline
$23$ & $1.1424\ldots$ & $\frac{617}{1587}$ & $\frac{617}{1587}$ & $\frac{52}{69}$ & $16$ & \ttt{00ECB6} & \ttt{1C312A} \\
$23$ & $1.1424\ldots$ & $\frac{617}{1587}$ & $\frac{617}{1587}$ & $\frac{52}{69}$ & $16$ & \ttt{08643E} & \ttt{14B9A2} \\\hline
$24$ & $1.1296\ldots$ & $\frac{19}{54}$ & $\frac{19}{54}$ & $\frac{7}{9}$ & $16$ & \ttt{032695} & \ttt{03CE6A} \\\hline
$25$ & $1.1546\ldots$ & $\frac{721}{1875}$ & $\frac{721}{1875}$ & $\frac{1444}{1875}$ & $32$ & \ttt{003B3CB} & \ttt{04E50A2} \\
$25$ & $1.1546\ldots$ & $\frac{721}{1875}$ & $\frac{721}{1875}$ & $\frac{1444}{1875}$ & $32$ & \ttt{01FAE32} & \ttt{0C42A69} \\\hline
$26$ & $1.1360\ldots$ & $\frac{217}{507}$ & $\frac{217}{507}$ & $\frac{359}{507}$ & $32$ & \ttt{042347C} & \ttt{0A6B813} \\\hline
$27$ & $1.1545\ldots$ & $\frac{833}{2187}$ & $\frac{833}{2187}$ & $\frac{188}{243}$ & $32$ & \ttt{0A109EC} & \ttt{3E2ACD6} \\\hline
$28$ & $1.1326\ldots$ & $\frac{53}{147}$ & $\frac{53}{147}$ & $\frac{227}{294}$ & $32$ & \ttt{071DBB5} & \ttt{3B82B6D} \\\hline
$29$ & $1.1625\ldots$ & $\frac{953}{2523}$ & $\frac{953}{2523}$ & $\frac{660}{841}$ & $32$ & \ttt{0089E14E} & \ttt{064BADE8} \\
$29$ & $1.1625\ldots$ & $\frac{953}{2523}$ & $\frac{953}{2523}$ & $\frac{660}{841}$ & $32$ & \ttt{013A6A2D} & \ttt{021C14EC} \\\hline
$30$ & $1.1288\ldots$ & $\frac{281}{675}$ & $\frac{281}{675}$ & $\frac{481}{675}$ & $32$ & \ttt{02A6CF21} & \ttt{10164AC7} \\\hline
$31$ & $1.1394\ldots$ & $\frac{1081}{2883}$ & $\frac{1081}{2883}$ & $\frac{2204}{2883}$ & $32$ & \ttt{067E2CAB} & \ttt{13AF5F63} \\\hline
$32$ & $1.1041\ldots$ & $\frac{1}{3}$ & $\frac{1}{3}$ & $\frac{37}{48}$ & $32$ & \ttt{0A363905} & \ttt{11D21EDD} \\
$32$ & $1.1041\ldots$ & $\frac{1}{3}$ & $\frac{1}{3}$ & $\frac{37}{48}$ & $32$ & \ttt{1EDD11D2} & \ttt{39050A36} \\\hline
$33$ & $1.1303\ldots$ & $\frac{1217}{3267}$ & $\frac{1217}{3267}$ & $\frac{2476}{3267}$ & $32$ & \ttt{00C03B652} & \ttt{0B5B9C517} \\\hline
$34$ & $1.1107\ldots$ & $\frac{353}{867}$ & $\frac{353}{867}$ & $\frac{610}{867}$ & $32$ & \ttt{00598B0ED} & \ttt{0DEE9382B} \\\hline
$35$ & $1.1746\ldots$ & $\frac{1361}{3675}$ & $\frac{1361}{3675}$ & $\frac{2956}{3675}$ & $32$ & \ttt{00963532E} & \ttt{1E280FB3B} \\\hline
$36$ & $1.1872\ldots$ & $\frac{29}{81}$ & $\frac{29}{81}$ & $\frac{403}{486}$ & $16$ & \ttt{1D603A324} & \ttt{7190953ED} \\
$36$ & $1.1872\ldots$ & $\frac{29}{81}$ & $\frac{29}{81}$ & $\frac{403}{486}$ & $16$ & \ttt{1DA30602B} & \ttt{7EACA6F12} \\\hline
$37$ & $1.1796\ldots$ & $\frac{1513}{4107}$ & $\frac{1513}{4107}$ & $\frac{3332}{4107}$ & $16$ & \ttt{00035AC726} & \ttt{0636C1F2AA} \\
$37$ & $1.1796\ldots$ & $\frac{1513}{4107}$ & $\frac{1513}{4107}$ & $\frac{3332}{4107}$ & $16$ & \ttt{0223D0E7AE} & \ttt{04164BD222} \\\hline
$38$ & $1.1209\ldots$ & $\frac{433}{1083}$ & $\frac{433}{1083}$ & $\frac{781}{1083}$ & $32$ & \ttt{0210F456C9} & \ttt{1A2E842E73} \\\hline
$39$ & $1.1477\ldots$ & $\frac{1673}{4563}$ & $\frac{1673}{4563}$ & $\frac{132}{169}$ & $32$ & \ttt{0019E3352D} & \ttt{07819B4CAA} \\
$39$ & $1.1477\ldots$ & $\frac{1673}{4563}$ & $\frac{1673}{4563}$ & $\frac{132}{169}$ & $32$ & \ttt{166AAF0FCC} & \ttt{195A5FFCC3} \\\hline
$40$ & $1.1033\ldots$ & $\frac{1}{3}$ & $\frac{1}{3}$ & $\frac{77}{100}$ & $16$ & \ttt{0033C66A5A} & \ttt{0F03369955} \\
$40$ & $1.1033\ldots$ & $\frac{1}{3}$ & $\frac{1}{3}$ & $\frac{77}{100}$ & $16$ & \ttt{33F0F55669} & \ttt{3CC005A566} \\\hline
$41$ & $1.1273\ldots$ & $\frac{1841}{5043}$ & $\frac{1841}{5043}$ & $\frac{3844}{5043}$ & $32$ & \ttt{00033C66A5A} & \ttt{00F03369955} \\\hline
\end{tabular}
\end{center}
\end{table}%
\setcounter{table}{2}%
\begin{table}[!ht]
\caption{(continued) Lowest Limiting Pursley-Sarwate Criterion among Seed Pairs that Have the Lowest Limiting Autocorrelation Demerit Factor}
\vspace{-0.9mm}
\begin{center}
\begin{tabular}{|c|lccc|c|c|c|}
\hline
\footnotesize{seed}  & \multicolumn{4}{|c|}{limiting values as $n\to\infty$} & \footnotesize{orbit} & \multicolumn{2}{c|}{{\footnotesize seeds}} \\
\footnotesize{length} & \multicolumn{1}{c}{\tiny{$\PSC(f_n,g_n)$}} & {\tiny $\ADF(f_n)$} & {\tiny $\ADF(g_n)$} & {\tiny $\CDF(f_n,g_n)$} & \footnotesize{size} & {\footnotesize $f_0$} & {\footnotesize $g_0$}  \\
\hline\hline
$42$ & $1.1171\ldots$ & $\frac{521}{1323}$ & $\frac{521}{1323}$ & $\frac{319}{441}$ & $32$ & \ttt{01A5024EE35} & \ttt{0978ED8A2B1} \\\hline
$43$ & $1.1474\ldots$ & $\frac{2017}{5547}$ & $\frac{2017}{5547}$ & $\frac{4348}{5547}$ & $16$ & \ttt{19685754CBC} & \ttt{3433FDFA1E6} \\
$43$ & $1.1474\ldots$ & $\frac{2017}{5547}$ & $\frac{2017}{5547}$ & $\frac{4348}{5547}$ & $16$ & \ttt{1BC85754E1C} & \ttt{3693FDFA346} \\\hline
$44$ & $1.1639\ldots$ & $\frac{125}{363}$ & $\frac{125}{363}$ & $\frac{595}{726}$ & $16$ & \ttt{0BE7818814C} & \ttt{67D44D4B285} \\\hline
$45$ & $1.1675\ldots$ & $\frac{2201}{6075}$ & $\frac{2201}{6075}$ & $\frac{4892}{6075}$ & $32$ & \ttt{01A1CE2AEC0D} & \ttt{097C919FF6BC} \\\hline
$46$ & $1.1109\ldots$ & $\frac{617}{1587}$ & $\frac{617}{1587}$ & $\frac{382}{529}$ & $32$ & \ttt{00F34581196A} & \ttt{1EF59BD8C174} \\\hline
$47$ & $1.1475\ldots$ & $\frac{2393}{6627}$ & $\frac{2393}{6627}$ & $\frac{5212}{6627}$ & $16$ & \ttt{0006E529E49C} & \ttt{363960F91AAA} \\
$47$ & $1.1475\ldots$ & $\frac{2393}{6627}$ & $\frac{2393}{6627}$ & $\frac{5212}{6627}$ & $16$ & \ttt{0886CF03E414} & \ttt{3EB94AD31A22} \\\hline
$48$ & $1.1643\ldots$ & $\frac{73}{216}$ & $\frac{73}{216}$ & $\frac{119}{144}$ & $16$ & \ttt{235B4408706E} & \ttt{235B45778F91} \\
$48$ & $1.1643\ldots$ & $\frac{73}{216}$ & $\frac{73}{216}$ & $\frac{119}{144}$ & $16$ & \ttt{3185A800F4D9} & \ttt{3185AABF0B26} \\\hline
$49$ & $1.1690\ldots$ & $\frac{2593}{7203}$ & $\frac{2593}{7203}$ & $\frac{5828}{7203}$ & $32$ & \ttt{007F570D131A6} & \ttt{0721128716E8A} \\
$49$ & $1.1690\ldots$ & $\frac{2593}{7203}$ & $\frac{2593}{7203}$ & $\frac{5828}{7203}$ & $32$ & \ttt{0220DA1E50EEE} & \ttt{0EEB06978B777} \\\hline
$50$ & $1.1093\ldots$ & $\frac{721}{1875}$ & $\frac{721}{1875}$ & $\frac{453}{625}$ & $32$ & \ttt{092A07192BF8E} & \ttt{1F8C92BF8E6D4} \\\hline
$51$ & $1.1560\ldots$ & $\frac{2801}{7803}$ & $\frac{2801}{7803}$ & $\frac{6220}{7803}$ & $16$ & \ttt{003CBF96B9CCE} & \ttt{133641E5434AA} \\
$51$ & $1.1560\ldots$ & $\frac{2801}{7803}$ & $\frac{2801}{7803}$ & $\frac{6220}{7803}$ & $16$ & \ttt{023615B4134EE} & \ttt{113CEBC7E9C8A} \\\hline
$52$ & $1.1863\ldots$ & $\frac{1}{3}$ & $\frac{1}{3}$ & $\frac{865}{1014}$ & $16$ & \ttt{00C3CC8A65695} & \ttt{03C0CFB9969AA} \\
$52$ & $1.1863\ldots$ & $\frac{1}{3}$ & $\frac{1}{3}$ & $\frac{865}{1014}$ & $16$ & \ttt{30F3FC455A566} & \ttt{33F0FF76A9A59} \\\hline
\end{tabular}
\end{center}
\end{table}%
\renewcommand{\arraystretch}{1.0}%

One can compare the results on Table \ref{Mary} (which reports the minimum limiting Pursley-Sarwate Criterion over all pairs of stems of a given length) with the results of Table \ref{Andrew} (which reports the minimum limiting Pursley-Sarwate Criterion only over pairs of stems whose limiting autocorrelation demerit factors equal the minimum value for that length).
In some cases (lengths $13$, $15$, $19$, $20$, $21$, and $23$ through $28$) the limiting Pursley-Sarwate Criterion reported on Table \ref{Mary} is lower.
We were able to go to much greater lengths in Table \ref{Andrew} because the computational burden is greatly reduced when we restrict the calculations of crosscorrelation properties to only those pairs of seeds that produce minimum asymptotic autocorrelation demerit factors.
The lowest asymptotic value for the Pursley-Sarwate Criterion we discovered was $331/300=1.10333\ldots$ for $\ell=40$ on Table \ref{Andrew}.
Our construction can be compared to randomly selected long binary sequences, which typically have Pursley-Sarwate Criterion of about $2$, and to high-performance sequence pairs constructed from finite field characters that have asymptotic Pursley-Sarwate Criterion of $7/6$ (see \cite[\S II.E, \S IV.D]{Katz} and \cite[eq.(6)]{Boothby-Katz}).

Our computations of limiting crosscorrelation demerit factors in Tables \ref{Mary} and \ref{Andrew} use a fast Fourier transform algorithm to speed up the convolutions (Laurent polynomial multiplications) that appear in the formula for asymptotic crosscorrelation demerit factor in Theorem \ref{Vivian}.  Because these calculations are performed using floating point arithmetic, there are small rounding errors.  We checked that the approximate values of the norms and integrals in our formula were always very close to integers: all discrepancies were less than $2\cdot 10^{-11}$.

\section*{Acknowledgments}

This research was done using computing resources provided by the Open Science Grid \cite{OSG-1,OSG-2}, which is supported by the National Science Foundation award 1148698, and the U.S. Department of Energy's Office of Science.
The authors thank Balamurugan Desinghu and Mats Rynge, who helped them set up the calculations on the Open Science Grid.
The authors thank the anonymous reviewers for their comments on the manuscript, many of which have helped improve this paper.

\end{document}